\documentclass[prodmode,permissions]{acmsmall-ec16}

\usepackage[ruled]{algorithm2e}

\SetArgSty{textrm}  
\SetAlFnt{\small}
\SetAlCapFnt{\small}
\SetAlCapNameFnt{\small}
\SetAlCapHSkip{0pt}
\IncMargin{-\parindent}

\usepackage{amsmath,amsfonts,amssymb,bbm} 
\usepackage[numbers,sort&compress]{natbib} 
\usepackage[lofdepth,lotdepth]{subfig}

\newcommand{\prob}{{\bf \mbox{\bf Pr}}}

\begin{document}

\markboth{D. T. Lee, A. Goel}{Towards large-scale deliberative decision-making: small groups and the importance of triads}

\title{Towards large-scale deliberative decision-making: small groups and the importance of triads}
\author{ASHISH GOEL and DAVID T. LEE\affil{Stanford University}}

\begin{abstract}
Though deliberation is a critical component of democratic decision-making, existing deliberative processes do not scale to large groups of people. Motivated by this, we propose a model in which large-scale decision-making takes place through a sequence of small group interactions. Our model considers a group of participants, each having an opinion which together form a graph. We show that for median graphs, a class of graphs including grids and trees, it is possible to use a small number of three-person interactions to tightly approximate the wisdom of the crowd, defined here to be the generalized median of participant opinions, even when agents are strategic. Interestingly, we also show that this sharply contrasts with small groups of size two, for which we prove an impossibility result. Specifically, we show that it is impossible to use sequences of two-person interactions satisfying natural axioms to find a tight approximation of the generalized median, even when agents are non-strategic. Our results demonstrate the potential of small group interactions for reaching global decision-making properties.
\end{abstract}


%



\begin{bottomstuff}
  This work was supported by the Army Research office Grant No. 116388, the Office of Naval Research Grant No. 11904718, and by the Stanford Cyber Initiative. David T. Lee was also supported by an NSF graduate research fellowship and a Brown Institute for Media Innovation Magic Grant.

  Author emails: \{ashishg,davidtlee\}@stanford.edu.
\end{bottomstuff}

\maketitle

\section{Introduction}

Recent years have seen an increase in the number of democratic innovations~\cite{Smith2009} aimed at increasing the participation of the public in policy-making. While the core argument of participatory democracy states that democracy is made more meaningful when citizens are directly involved~\cite{Pateman1979,Macpherson1977,FW03}, it has also been argued that participatory practices foster more engaged citizens~\cite{FW03} and can better harness the wisdom of the crowds~\cite{Landemore2012}. One class of these processes emphasizes deliberation as an essential component to democratic decision-making. Proponents argue that deliberation grants legitimacy to the democratic process~\cite{Gutmann2004}, can produce outcomes superior to those resulting from mere voting~\cite{F11}, and may be a means through which one can mediate conversations among polarized participants~\cite{Menkel-Meadow2011}. The latter is particularly desirable in light of the increasing polarization that has been observed in politics and society at large~\cite{PR84,MPR06}.

A major challenge for increasing participation is the problem of effectively scaling up to thousands or millions of people. Yet, while there have been many efforts at scaling up voting to large numbers of proposals~\cite{Conitzer2005,Service2012,Lu2011,Lee2014}, the problem of scaling up deliberation has been largely untouched. One part of the problem is that humans face fundamental cognitive and psychological barriers that prevent effective communication in large groups. For instance, it is known that humans can only hold a small number of concepts in their working memory~\cite{Miller56}, that there are fundamental limits to effective group sizes~\cite{Dunbar1992}, and that humans can become more disengaged when in a large group~\cite{LWH79}. 

Existing initiatives such as The Deliberative Poll~\cite{FLJ00} and the 21st Century Town Hall Meeting~\cite{LB02}, while successful on many counts, have typically operated at the scale of hundreds, or at most, thousands of participants. Even for these events, a major component of these processes involve further breaking down the gathered participants into small group discussions.

Motivated by this observation, this paper considers whether large-scale decision-making can take place through a sequence of small group interactions. Essentially, we take the view that the deliberative decision-making process can be, at its core, the interaction of many separate but intertwined deliberations. The use of small groups circumvents communication barriers, changing the problem of scalability from that of increasing group size to that of decreasing the number of small group interactions necessary. Thus, our use of small groups for scaling up deliberation can be thought of as analogous to the use of pairwise comparisons for scaling up voting when the number of proposals is very large~\cite{Lee2014}.
From a practical perspective, scaling deliberation via small groups would not have been feasible even 10 years ago. Today, however, advances in communication and Internet-based technologies such as video conferencing have made it trivial for arbitrary small groups of individuals to deliberate with one another, and moreover, to sequence these deliberations based on outcomes of prior small groups.  

In our model, participants each have an opinion which together form a graph. Participants prefer opinions that are closer to their own. Our goal is to find the generalized median using a sequence of small group interactions, where the generalized median is defined as the point minimizing sum of distances to the initial opinions.
Each small group is required to come to a small group decision, which is modeled by us in two distinct ways. For our impossibility results, we assume that agents do not behave strategically, and that the small group decision is therefore a (possibly stochastic) function of the opinions of the small group members. We assume that this function satisfies natural axioms. For our positive results, we allow strategic agents. This means that the small group decision is the result of the negotiation strategies chosen by the individuals in that group. Under this framework, our main findings uncover a sharp dichotomy among small group processes:
\begin{itemize}
  \item First, we obtain a surprising impossibility result for \emph{dyadic decision-making} even for non-strategic agents. We show that under natural axioms, no sequence of two-person interactions can accurately estimate the generalized median. This holds even for the simple case when opinions lie on a line.
  \item In contrast, for \emph{triadic decision-making}, we obtain a strong positive result even for strategic agents. We prove that when the small group decision is decided through majority rule dynamics, a small number of three-person interactions (an average of only $O(log^2 n)$ triads per person) finds a tight approximation of the generalized median. This holds not only when agents have opinions from a line, but for {\it every} graph for which every three individuals in the graph have a majority rule equilibrium (we prove this is identical to a class of graphs known as median graphs). The mechanism we give is shown to achieve our results under a subgame perfect Nash equilibrium.
\end{itemize}

Figure \ref{fig:triad} illustrates our mechanism for the simple case where the opinions of the voters lie on a line. Given a triad, our mechanism for this case reduces to each member voting among the other two. This results in the median participant winning two votes, and getting increased weight in later rounds. Details are described in later sections. To show our positive results for triadic decision-making, we first show that our mechanism induces a behavior in which every triad reports its true generalized median; hence our game-theoretic result is also valid if we postulate an small group decision dynamic where triads agree on their generalized median.

\begin{figure}[t]
  \centerline{\includegraphics[width=.5\textwidth]{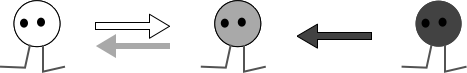}} 
  \caption{\small For the special case of triads on a line, our mechanism can be reduced to having each individual vote between the other two participants. This results in the median participant receiving two votes.}
   \protect\label{fig:triad}
\end{figure}

These findings, put together, illustrate the importance of triads for small group decision-making. The fact that only an extremely small number of triads is required is a positive indicator of the feasibility of achieving deliberative decision-making at scale. The fact that it is not achievable with dyadic interactions is surprisingly initially, but makes intuitive sense after the fact: when participating in a two-person interaction, it is impossible to protect against an extreme outlier. If two individuals are to be treated fairly, an arbitrarily unreasonable individual could greatly affect the decision made. On the other hand, three-person interactions can essentially ignore single outliers through simple processes such as majority rule.

Before moving on, we also want to put our results in the broader context of deliberative democracy. Our model is focused on large-scale decision-making through small group interactions, which we believe will play an important part of deliberative systems due to the benefits of back-and-forth dialogue. However, our results do {\it not} account for an information-based evolution of beliefs. Our main theorem (Theorem \ref{thm:triadic}) essentially states that, if groups of three can find their generalized median, then we can find the global generalized median {\it of initial opinions} by composing a small number of three-person decisions. If beliefs change, then one would really want to find the generalized median of {\it new opinions} after sufficient information exchange has occured. From a political theorist's perspective, deliberative democracy must involve a reasoned exchange of information, and not just negotiation~\cite{Gutmann2004}. Therefore, a critical step for the future is to consider models of decision-making where information is taken into account. It is also important to note that, while deliberation was traditionally viewed as an exchange of arguments among individuals, theorists have recently expanded that definition to encompass richer dynamics. In this view, political debates, news media, or dialogue over a social network (that has the potential to influence decision-making outcomes) can also be viewed as a part of a deliberative system. We believe that studying models for decision-making that incorporate richer, deliberative dynamics is an important question for the future.

\section{Model}

\subsection{Participant opinions and median graphs}

In our model, $n$ participants have opinions $V = \{x_1, x_2, \ldots, x_n\}$ which are nodes of a graph $G = (V, E)$. This graph is a representation of the opinion space. Each participant prefers opinions that are closer to his own, where the distance $d(x, y)$ is defined as the shortest path distance between opinions $x$ and $y$ in the graph. The participants participate in an algorithmically defined sequence of small group interactions, each of which results in a {\it small group decision}. When a stopping criterion is reached, a function maps the sequence of small group decisions to a {\it final consensus decision}.

A common example of an opinion space is to consider opinions on a line representing the liberal-conservative axis. This line could, for example, lie in a high-dimensional space (e.g. specific budget or policy proposals), so long as participants preferences are determined by distances on the line (how conservative or liberal the proposals are).

For our positive results on scaling decision-making using a sequence of triads (Theorem \ref{thm:triadic}), our results are only true for graphs where every three individuals in the graph have a majority rule equilibrium (also known as Condorcet winner), i.e. there exists an opinion that would receive two out of the three votes against any other opinion. This turns out to coincide to a class of graphs called median graphs (we prove this in Theorem \ref{thm:median-graph-characterization}). Median graphs include several common classes of graphs such as trees, grids, and squaregraphs. As a toy example, if one were deciding on ice-cream flavors, one  might represent the opinions as a two-dimensional grid graph where the two axes correspond to discrete levels of sweetness and fattiness. If people's preferences for ice-cream flavors are roughly determined independently by sweetness and fattiness (or any finite number of dimensions), then this could be represented by a median graph. However, there are many graphs which are not median graphs. For instance, any graph with a triangle cannot be a median graph.


\subsection{The generalized median as the wisdom of the crowd}

\begin{figure}
  \centerline{\includegraphics[width=.25\textwidth]{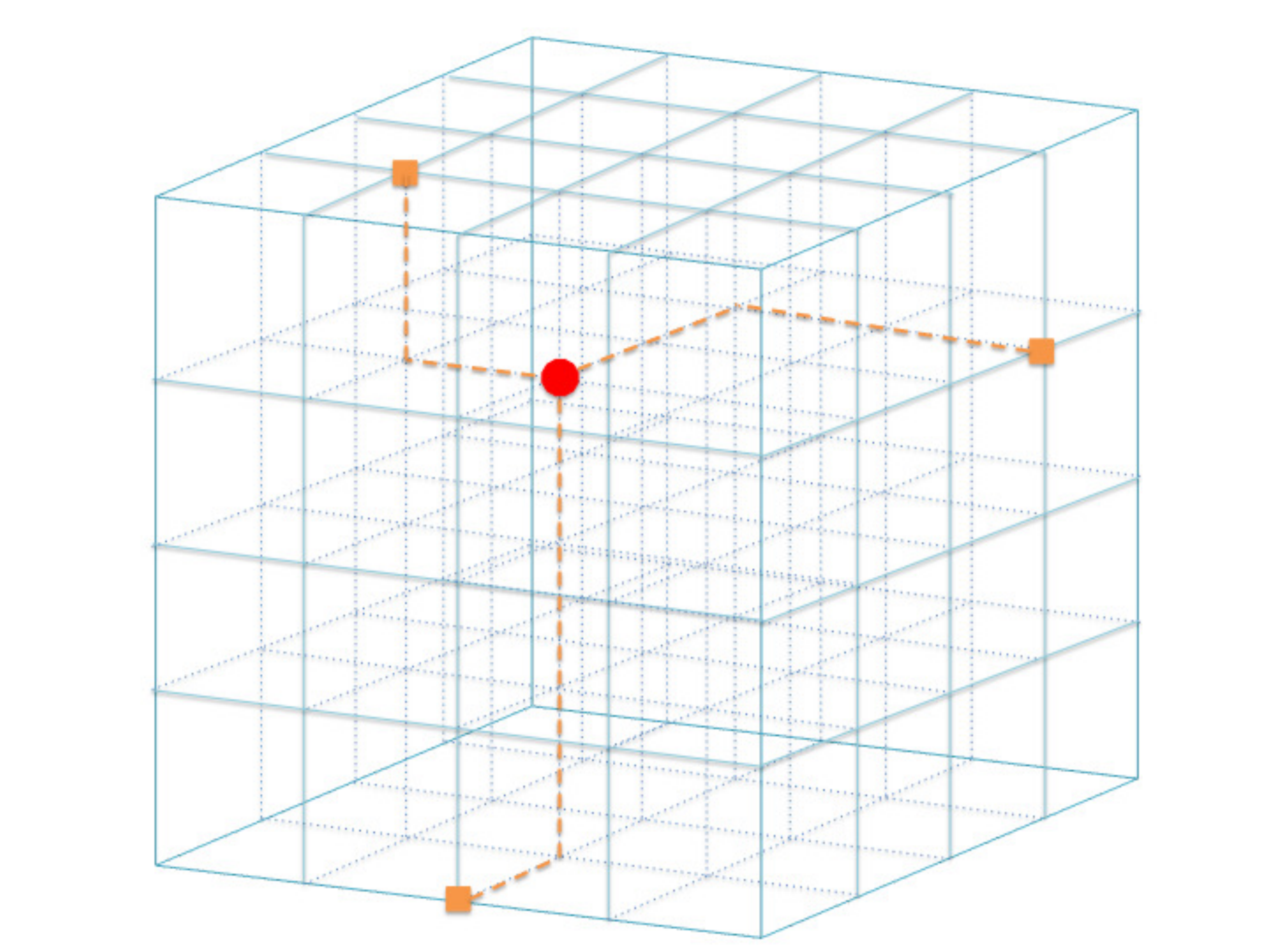}}
  \caption{\small The red circle is the generalized median of the three orange squares.}
  \label{fig:median}
\end{figure}

Our objective is to find the wisdom of the crowd, here defined to be the generalized median $x^*$, the opinion which minimizes the sum of distances to all of the participant opinions (see Figure \ref{fig:median}). Formally, $x^* = \arg\min_{x \in V} D(x)$, where $D(x) = \sum_{i=1}^n d(x, x_i)$. The decision-making process achieves a $(1+\epsilon)$-approximation of the generalized median if the final consensus decision is a point $\hat{x}$ for which $D(\hat{x}) \leq (1+\epsilon)D(x^*)$, i.e. $\hat{x}$ \emph{almost} minimizes the sum of distances. We seek sequences of small group interactions for which $\epsilon \to 0$ as the number of participants increases ($n \to \infty$).

The concept of the wisdom of the crowd~\cite{Surowiecki2005} originated in 1907 from Sir Francis Galton's paper~\cite{Galton1907} describing a competition at a carnival in which participants could guess the weight of an ox. As people made their estimates, Galton recorded them and observed that the median, which he called the {\it Vox Populi} or voice of the people, was remarkably close to the correct answer. Based on this, he hypothesized that an appropriate aggregation of a crowd's preferences can produce an extremely accurate estimate.

Much of the following theoretical work on the wisdom of the crowd has chosen to study the mean rather than the median~\cite{Golub2010}. We follow Galton, who stated that he preferred the median over the mean due to its robustness to outliers and the fact that the median is the Condorcet winner, the opinion which is preferred by at least half of the participants to any other opinion~\cite{Hooker1907,Galton1907b}. Since we consider opinions that lie in high-dimensional metric spaces, we take the wisdom of the crowd to be the \emph{generalized} median, which is also robust to outliers~\cite{Lopuhaa1991} and closely related to the Condorcet winner~\cite{Bandelt1984,Saban2012}. We also note that recent experimental studies also support the position that the median is more accurate than the mean~\cite{Lorenz2011}.

\subsection{Large-scale decision-making via small groups}

We will use an {\it LDSG process} to refer to a large-scale decision-making process conducted through small group interactions. In our process, each participant starts with $k$ tokens, which can be thought of as $k$ votes. Unless otherwise specified, we will assume $k=1$. A small group of participants is chosen by a (possibly stochastic) {\it small group selection function} $f$ to deliberate with one another, with the requirement that the members must come to a small group decision on a single participant to be their representative. This participant, who does not need to be one of the group members, is then given one token from each of the group members. We use $S_t$ to represent the $t$-th small group and $g_t$ to represent the small group decision that they come to.
Another small group is chosen (again, by the function $f$), which can be dependent on the past history of small group decisions. This process is repeated until one participant has all the tokens. This participant's opinion is returned as the final consensus decision $\hat{x}$. Formally,

\begin{definition}\label{def:ldsg-process}
  An LDSG process $L$ is a triple $L = (V, f, g)$, where $V$ is the set of participants $V = \{x_1, x_2, \ldots, x_n\}$, $f$ is the {\it small group selection function} which maps the sequence of past small groups and small group decisions to the next small group, and $\{g_t, t \geq 1\}$ are the small group decisions chosen by the $t$-th small group. Let $S_t$ denote the small group chosen to interact in the $t$-th round and let $y_i^t$ denote the participant holding the $i$-th token after the $t$-th round. Let $T$ be the round at which the process ends, and $\hat{x}$ denote the final consensus decision made. Then,
 \begin{align*}
   y_i^0 &= x_i,\\
   y_i^{t+1} &= \begin{cases}
     g_{t+1} &\text{ if $y_i^t \in S_{t+1}$}\\
     y_i^{t}   &\text{ otherwise}
   \end{cases},\\
   S_{t+1} &= f(\langle S_1, g_1, S_2, g_2, \ldots, S_t, g_t \rangle),\\
   T &= \min\{ t \mid y_1^t = y_2^t = \cdots = y_n^t \},\\
   \hat{x} &= y_1^T
 \end{align*}
\end{definition}

\subsection{Small group decision-making}

In our model's full generality, the small group decision $g_t$ can be modeled in two distinct ways. For non-strategic agents, we represent it by a {\it small group decision function} $g : S \to \mathcal{P}_V$, which maps any small group $S \subseteq V$ to a probability distribution $\mathcal{P}_V$ over the set of opinions $V$, i.e. $g_t = g(S_t)$.
For strategic agents, the decision of a small group is taken to be the output of a mechanism that the small group participates in. This output is a function of the strategic decisions that the small group participants choose to make. Since participants can base their strategies on the entire history of the LDSG game up to that time, $g_t$ is not necessarily a function solely of $S_t$.

For our dyadic impossibility results, we will consider non-strategic agents and show that no function $g(S_t)$ exists which simultaneously satisfies natural axioms and also results in tight approximation of the generalized median. In our result on triadic interactions, we will consider small groups which make their decision through a simple majority rule process. In the non-strategic case, $g(S_t)$ is assumed to be the majority rule equilibrium for that small group (described further in Section \ref{sec:triadic}). In the strategic case, we will give a formal representation of the entire sequence of small group majority rule processes as an infinite-horizon extensive form game (Section \ref{sec:triadic-strategic}). 

\section{Outline}

After describing further related work (Section \ref{sec:related-work}), the remainder of this paper will be structured as follows. Section \ref{sec:dyadic} will present our impossibility result on dyadic interactions. Section \ref{sec:triadic} contains our results on scaling deliberative decision-making with three-person {\it non-strategic} interactions. This section will formally define the majority rule process and equilibrium that motivates our model of the small group decision function $g(S_t)$. It will then use this to show that $O(n\log^2n)$ triads suffice to find a tight approximation of the generalized median for a large class of graphs. Section \ref{sec:triadic-strategic} expands on the prior section by considering the case when agents are strategic and show that we can achieve the same results under a subgame perfect Nash equilibrium. Finally, we conclude in Section \ref{sec:discussion} with a brief discussion of future directions.

\section{Related Work}\label{sec:related-work}

The problem of scaling up decision-making is a well-studied problem in the context of voting and preference elicitation. In this context, one is typically trying to approximate or calculate the output of a social choice function while only eliciting small amounts of information from voters (e.g. through pairwise comparisons)~\cite{Conitzer2005,Service2012,Lu2011,Lee2014}.
Our paper can be viewed as an attempt to create a thread of research mirroring preference elicitation, but for deliberative processes. Thus far, computational social choice has primarily viewed decision-making and preference elicitation from the perspective of efficiency, accuracy, and strategic issues. We propose that deliberation is a valuable new dimension to consider in social choice and provide a small step in this direction. We discuss this direction more in the concluding section. 

Median graphs have been studied before in the context of voting. For instance, it is known that for median graphs, the Condorcet winner is strongly related to the generalized median~\cite{Bandelt1984,Wendell1981,Saban2012}. Nehring and Puppe~\cite{Nehring2007} show that any single-peaked domain which admits a non-dictatorial and neutral strategy-proof social choice function is a median space. Clearwater et. al. also showed that any set of voters and alternatives on a median graph will have a Condorcet winner (their full result is stronger than this)~\cite{CPS2015}.

A similar high-level triadic decision-making process was proposed in a prior paper by the authors~\cite{Goel2012} and analyzed for the restricted case (described in Figure \ref{fig:triad}) of opinions on a line. In this paper, we consider richer small group decision dynamics such as majority rule. We also consider spaces more complex than the line and provide a general framework for analyzing deliberative decision-making that enables us to prove an impossibility result for dyads.

To our knowledge, we are not aware of other literature on algorithmic or mathematical models for scaling deliberation. However, the problem of scaling deliberation has been discussed in the political science community. Deliberation was initially conceptualized as an exchange of arguments among a small group of rational individuals. Recent developments acknowledge that deliberation should be thought of in richer ways, including as an activity that takes place at the system level, among groups.~\cite{MP12}. In other words, it has recently become possible to think of deliberation as a task that can be broken down into various components and performed at different social levels and by large numbers of individuals.
There have also been several practical initiatives for scaling deliberation such as Deliberative Polling~\cite{FLJ00}, in which a single representative sample of participants is brought together to deliberate, and the 21st Century Town Hall Meeting~\cite{LB02}, in which the entire set of participants are divided into tables of size $10$ for a single round of deliberation followed by voting.

The study of majority rule dynamics, which is an important part of deliberative decision-making, has also been a long-studied problem. A particular relevant experimental result analyzed the majority rule dynamic in groups of five and showed that the solution concept that performed best (out of 16 considered) was the majority rule equilibrium~\cite{Plott1978}.
Our work is a natural next step of this observation applied to the goal of scaling deliberative decision-making. Namely, given experimental evidence for the ability of small groups to come to consensus on the majority rule equilibrium, how can we use these small groups to make good decisions for larger groups?

Finally, we also note that our work has strong connections to opinion formation dynamics. Deliberative decision-making processes are essentially opinion formation dynamics for which one can algorithmically choose the sequence of interactions, rather than having them decided by a given social network~\cite{D74} or a nature-induced flocking process~\cite{T84}. Most models of opinion formation use a weighted averaging dynamic and consider the mean rather than the median. We note that our main result on the dichotomy between triadic and dyadic decision-making can be adapted for opinion formation dynamics. Namely, for any opinion formation dynamic to have a chance at approximating the median, individuals should talk to at least two people before they update their opinions. Otherwise, the group consensus will be too easily influenced by extreme opinions.

\section{Dyads cannot tightly approximate the generalized median}\label{sec:dyadic}

We now prove our impossibility result for dyadic decision-making. This result holds for all small group decision functions $g$ satisfying an axiom we call local consistency. Informally, local consistency captures the fact that the decision made among a set of non-strategic participants should only depend on the relationship structure of the participant opinions to each other. In other words, if one set of opinions is a ``translation'' of a second set of opinions, then the small group decision of the second set should be the same ``translation'' of the small group decision of the first set. This is depicted visually in Figure \ref{fig:local}.

More formally, we use the notion of convex sets and convex hulls to capture the {\it opinion structure} of a set of participants. A convex set is any set such that for any $x, y$ in the set, all shortest paths between $x$ and $y$ also lie in the set. The convex hull of $S$, denoted as $C_S$, is the smallest convex set containing $S$.



\begin{figure}
  \centerline{\includegraphics[width=.5\textwidth]{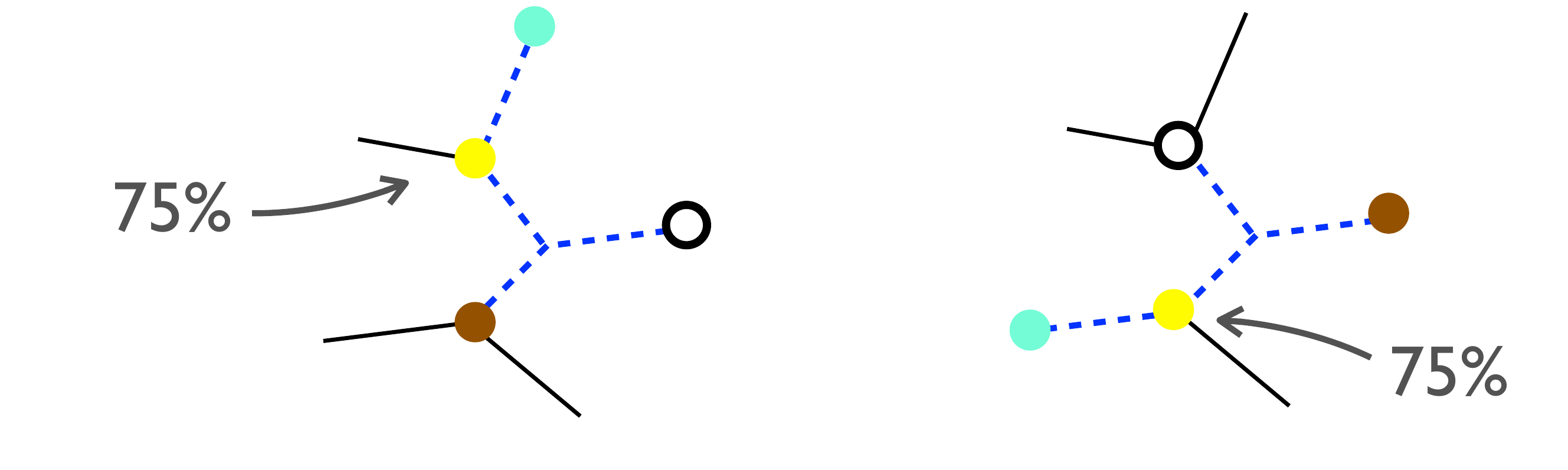}} 
  \caption{\small The green, brown, and white participants in the two separate graphs have the same opinion structure since their positions relative to their convex hulls (the dotted blue lines) are identical. In this example, a locally consistent small group decision function would have to assign the same probability to the event in which the participants choose the yellow participant in either graph.}
   \protect\label{fig:local}
\end{figure}

\begin{definition}
  $S$ and $S'$ are said to have the same opinion structure under $\psi$ if $\psi : C_S \to C_{S'}$ is an isometry (distance-preserving map) between $C_S$ and $C_{S'}$ such that $S' = \psi(S)$.
\end{definition}

\begin{definition}
  The small group decision function $g$ is said to be locally consistent if, (1) for any $S$, $g(S) \in C_S$, and (2) for any $S$ and $S'$ with the same opinion structure under $\psi$, $g(S') = \psi(g(S))$.
\end{definition}
For equally influential participants, this definition encompasses well-known mathematical models in opinion formation such as DeGroot and flocking~\cite{D74,T84}.
We also note that, while we only focus on equally influential participants in this paper due to space contraints, our impossibility results also generalize to participants with varying degrees of influence (with an appropriate generalization of the model to account for an influence parameter).

\begin{theorem}\label{thm:pairwise}
  Consider the LDSG process $L = (V, f, g)$ where $f$ is defined in any way such that $\lvert S_t \rvert = 2$ for all $t$, and $g_t = g(S_t)$ where $g(\cdot)$ is any locally consistent function. Then there exists a configuration of the participant opinions on a line graph such that the final consensus decision $\hat{x}$ satisfies 
  \[\mathbb{E}[D(\hat{x})] \geq \left(\frac{9}{8}-o(1)\right)D(x^*)\]
  where $x^*$ is the generalized median.
\end{theorem}
\begin{proof}
  The proof is by constructing a counterexample (see Appendix \ref{appsec:dyadic-proof}).
\end{proof}

In other words, simple examples exist for which the error $\epsilon$ is always lower bounded by $1/8$, regardless of how many participants are involved or how many small group interactions take place. 
We pont out that a constant factor approximation is not satisfying since, in many scenarios, even the {\it worst} candidate is a constant factor approximation.


\section{$O(n\log^2n)$ triads tightly approximate the generalized median}\label{sec:triadic}

In contrast to the impossibility result concerning dyadic decision-making, we find that triadic decision-making under a simple majority rule process is able to find an extremely tight approximation of the generalized median while only requiring an average of $O(\log^2 n)$ triads per participant. 

The main step in the proof is to first assume that every triad decides on their (local) generalized median. We show that if this is true, then we can find a tight approximation of the (global) generalized median. This assumption is justified by showing that the generalized median is the equilibrium of the majority rule process for any small group coming from the preference domain we are considering (median graphs). This section will only consider non-strategic agents. Section \ref{sec:triadic-strategic} will extend this to the strategic case.

\subsection{Triadic decision-making}\label{sec:triadic-deliberation}

The small group selection function $f$ we use for selecting small groups is a simple randomized process. At each step, three tokens will be selected uniformly at random with replacement. The participants holding these tokens form the group $S_t$ at time $t$. 

Our result (Theorem \ref{thm:triadic}) in Section \ref{subsec:scalability-theorem} does not depend on any specific small group process, only on the assumption that every three-person interaction decides on the generalized median. This section, however, will motivate the generalized median by considering small groups which make their final decision through a specific process: majority rule (described below, and depicted in Figure \ref{fig:ldsg-game}). We model the small group decision function as the majority rule equilibrium and consider only those graphs for which a majority rule equilibrium exists for every group of three. For such graphs, it turns out that the majority rule equilibrium coincides with the generalized median. Thus, we have the following formalization.

\begin{definition}\label{def:tdd}
  Define triadic decision-making as an LDSG process $L = (V, f, g)$ where $g_t = g(S_t) = \text{generalized-median}(S_t)$. For round $t+1$, let $I_{t+1} = \{u, u', u''\}$ denote three tokens selected uniformly at random with replacement, i.e. $u, u', u'' \overset{i.i.d.}{\sim} U_n$, where $U_n$ denotes the uniform distribution over integers $1, 2, \ldots, n$. Then define $f$ such that, for all $t$,
  \begin{align*}
    S_{t+1} = \{ y_{u}^{t}, y_{u'}^{t}, y_{u''}^{t}\}
  \end{align*}
\end{definition}

\begin{figure}[t]
  \centerline{\includegraphics[width=.4\textwidth]{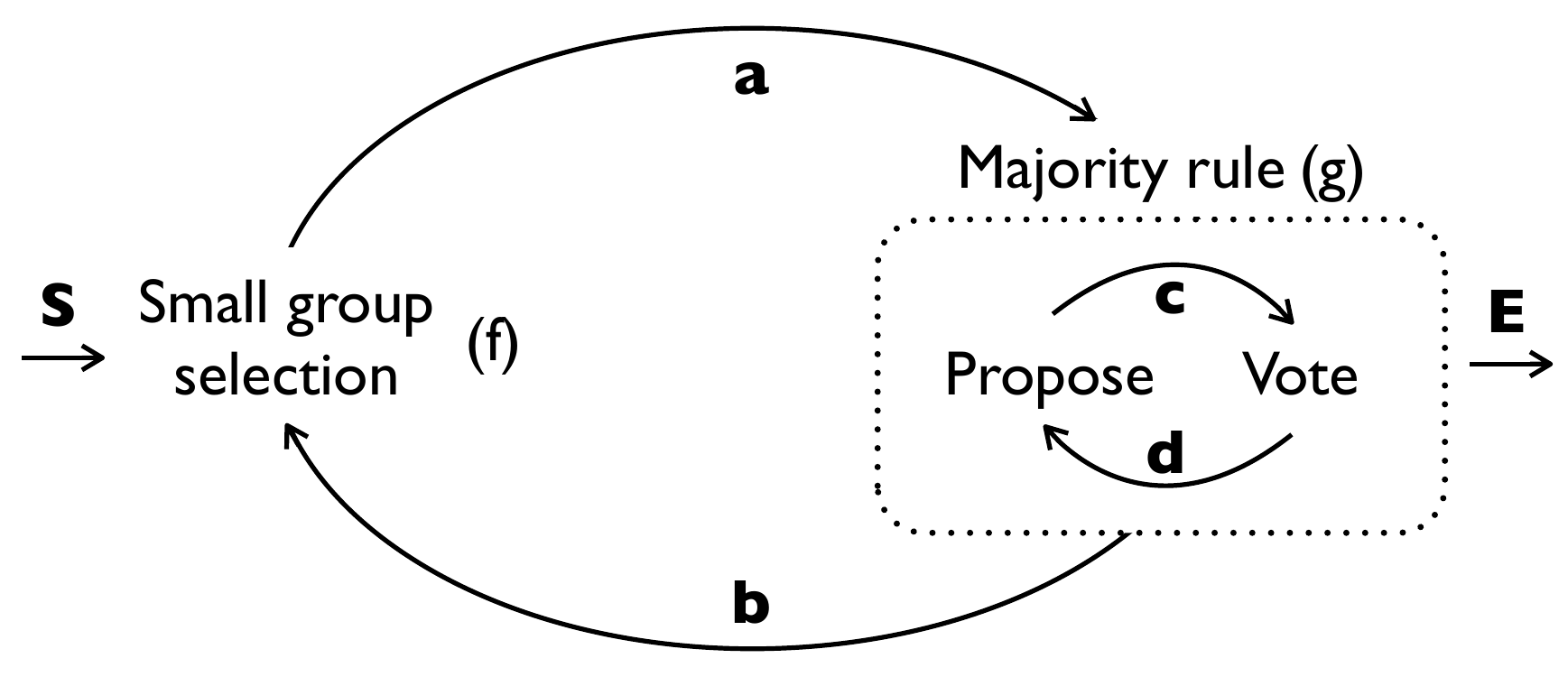}}
  \caption{\small A block diagram of triadic decision-making with majority rule. The process starts (S) with selecting a small group (``Small group selection''). After selecting a small group (a), the small group participates in the majority rule process which starts with a ``Propose'' action. After a proposal (c), all three participants ``Vote'' on the proposal. If there was no motion to end which was accepted by two of the three participants, then another proposal is taken (d). This repeats until a participant motions to end, and two of the three participants accept. If all tokens now belong to a single participant, the overall process ends (E). Otherwise, another small group is selected (b) and the new group of participants engage in the majority rule process. Extensive form representations of each block are sketched in Figure \ref{fig:ldsg-extensive-form}.}
  \label{fig:ldsg-game}
\end{figure}

Majority rule is a classical process for making a decision in a group. Majority rule occurs in rounds in which participants repeatedly suggest an alternate decision to some status quo. If at least two of the three participants vote for the alternate, then the status quo is replaced with the alternative. 
In practice, there are many variations in the way majority rule is implemented. For concreteness, we consider the following majority rule process for groups of three participants. 

\begin{itemize}
  \item \emph{(Initial)} At the beginning, two participants are chosen uniformly at random. One is chosen to be the {\it current winner} and the other is chosen to be the {\it proposer},
  \item \emph{(Propose)} At the start of each round, the proposer proposes an alternative winner,
  \item \emph{(Vote)} All three participants then simultaneously vote for either the current winner or the suggested alternative,
  \item \emph{(Update)} The proposal which received the majority of the votes becomes the new current winner. If the vote was unanimous, then the proposer remains the same. Otherwise, the participant who lost the vote becomes the new proposer.
  \item \emph{(Repeat)} Another round takes place starting from the Propose step,
  \item \emph{(Termination)} During any round, the proposer can propose to end the process instead of proposing an alternative. If he does, and if a majority votes to accept, then the process ends. Otherwise, the proposer is updated according to the same rule in the Update step and another round takes place starting from the Propose step.
\end{itemize}

A majority rule equilibrium (MRE), also known as the Condorcet winner, is a proposal which would beat any other proposal, i.e. always receives a majority of the three votes cast. Whenever this proposal is voted on, it will win. Therefore, it is a commonly accepted solution concept for the decision resulting from a majority rule dynamic.

For this paper, we will consider only graphs for which every group of three participants has a majority rule equilibrium. It turns out that this requirement coincides exactly with a class of graphs known as {\it median graphs}~\cite{Bandelt2008,Knuth2011} (see Theorem \ref{thm:median-graph-characterization}), which include common classes of graphs such as trees, grids, and squaregraphs. Moreover, for median graphs, the majority rule equilibrium of any three participants coincides with the generalized median. 

\begin{theorem}\label{thm:median-graph-characterization}
  Consider any unweighted and undirected graph $G$. Let $P$ denote the property in which every set of three nodes has a Majority Rule Equilibrium (Condorcet winner). Then,
  \[\text{$G$ satisfies $P$ if and only if $G$ is a median graph}\]
  Moreover, the Condorcet winner of any set of three nodes is the unique generalized median for that set.
\end{theorem}
\begin{proof}
  See Appendix \ref{appsec:median-graph-characterization} for the proof.
\end{proof}

Thus, we let $g_t = g(S_t) = \text{majority-rule-equilibrium}(S_t) = \text{generalized-median}(S_t)$. As is also mentioned in the proof of Theorem \ref{thm:strategic}, for any one of the three individuals (say, $x$), the generalized median is also the closest point to $x$ that lies between $y$ and $z$. Because of this, it is reasonable to suppose that the small group of individuals can find their equilibrium quickly.

\subsection{Scalability theorem for triadic decision-making}\label{subsec:scalability-theorem}

We now show that a small number of triadic interactions can tightly approximate the generalized median. We highlight the stark contrast between the approximation achieved by triadic decision-making, for which the error $\sqrt{c\log n/n}$ tends to $0$ as $n$ becomes large, and dyadic decision-making, for which the error is lower bounded by $1/8$.

\begin{theorem}\label{thm:triadic}
  Consider a triadic decision-making process $L = (V, f, g)$ where $n$ participants $V$ form a median graph. Then, with probability at least $\displaystyle 1-\frac{n\log n}{n^c}$, $c \geq 2$,
  \begin{align*}
    T = O(cn\log^2 n)\quad\text{and}\quad D(x^*) \leq &D(\hat{x}) \leq \left(1+O\left(\sqrt{\frac{c\log n}{n}}\right)\right)D(x^*)
  \end{align*}
\end{theorem}
\begin{proof}
A high-level understanding of our proof is that, for median graphs, the given urn dynamic can be reduced to several (simpler) coupled urn dynamics, represented by edges in $E^*$. $E^*$ (defined in Appendix \ref{appsec:median-graphs}) can be thought of roughly as ``independent dimensions'' of a median graph. Most of the proof is then tying the approximation factor or time of the overall urn process to those of the simpler coupled urn dynamics.

Our proof will be split into several lemmas. The lower bound $D(\hat{x}) \geq D(x^*)$ is trivial by the definition of $x^*$ as the generalized median. 
To prove the upper bound, we first rewrite $D(x)$ as,
\begin{align*}
D(x) = \sum_{i=1}^n d(x, x_i) = \sum_{i=1}^n \sum_{e\in E^*} \mathbbm{1}\{e \in I_{xx_i}\} = \sum_{e \in E^*} \lvert \{ i \mid e \in I_{xx_i}\} \rvert
\end{align*}
where $I_{xy}$ is the set of nodes lying on a shortest path between $x$ and $y$ (see Appendix \ref{appsec:definitions}). The second equality comes from a property of median graphs stating that $d(x, y) = \sum_{e\in E^*} \mathbbm{1}\{e \in I_{xy}\}$ (Lemma \ref{lem:winsets-4} in the Appendix). 
From this, it follows that
\begin{align*}
  \frac{D(\hat{x})}{D(x^*)} \leq \max_{e \in E^*} \frac{\lvert \{ i \mid e \in I_{\hat{x}x_i}\} \rvert}{\lvert \{ i \mid e \in I_{x^*x_i}\} \rvert}
\end{align*}
In Lemma \ref{lem:edge-approx}, we show that
\begin{align*}
  \prob\left[\frac{\lvert \{ i \mid e \in I_{\hat{x}x_i}\} \rvert}{\lvert \{ i \mid e \in I_{x^*x_i}\} \rvert} > \left(1+O\left(\sqrt{\frac{c\log n}{n}}\right)\right)\right] \leq \frac{1}{n^c},
\end{align*}
We can now achieve our upper bound via a union bound on all edges of $E^*$, noting that $\lvert E^* \rvert \leq \lvert E \rvert \leq n\log n$ (Lemma \ref{lem:winsets-3}).

To bound $T$, we first define $T_{e = (u, v)}$ as the first time $t$ at which all tokens $y_i^t$ lie in either $W_{uv}$ or $W_{vu}$ (informally, either closer to node $u$ or node $v$, see Definition \ref{def:winsets}). By Lemma \ref{lem:winsets-2}, we have
\begin{align*}
  T = \max_{e \in E^*}T_e
\end{align*}
We show in Lemma \ref{lem:edge-markov-chain} that $T_e = O(cn\log^2 n)$ with probability at least $1-\frac{1}{n^c}$. Then we can achieve our result $T = O(cn\log^2 n)$ via a union bound on all edges of $E^*$.
\end{proof}

When one considers grid graphs and trees, one can derive remarkably strong results using roughly the same proof techniques. Triadic decision-making is essentially able to find the {\it exact} generalized median with high probability, a very surprising result given the extreme randomness of the process.


\begin{theorem}
  Consider $n = 9k^2$ participants which make up a $k \times k$ grid, so that each point in the grid has $9$ participants. Assume $k$ is odd. Then with probability greater than 99.4\%, the winning participant $\hat{x}$ is the exact generalized median $x^*$ (the middle point of the grid) {\it regardless of how large $k$ is}.
\end{theorem}

\begin{theorem}
  Consider $n$ participants which make up a binary tree of height $h$. Then with probability greater than $1-4e^{-n/16}$, the winning participant $\hat{x}$ is either the root node (the generalized median) or one of its two children. This is again true {\it regardless of how large $k$ is}.
\end{theorem}
Even stronger results hold for higher-dimensional grids or higher degree trees.

\subsection{A simpler, restricted variant}\label{sec:restricted}

The downside of the above setting is that participants will need to take time to discover the opinions of other participants. An implementation of this in practice would require more discovery mechanisms and tools. An accurate analysis would require a model of this discovery process, which is out of the scope of this paper.

In this section, we consider a simpler setting in which the small group decision is restricted to one of the members of the triad itself, i.e. $g_t \in S_t$. It turns out that an extremely simple process exists for finding the majority rule equilibrium under this restriction: simply ask each participate to vote on one of the other two participants. The participant who receives the largest number of votes wins. If there is a tie, then there is no exchange in tokens. This dynamic was analyzed in \cite{Goel2012} and shown to obtain essentially similar results for the case of opinions on a line. Here, we extend these results to star nodes (Theorem \ref{thm:star-restricted}) and also give simulations indicating that it also works for trees and grids (Figure \ref{fig:simulations}). Extending this to all median graphs (and perhaps beyond) remains an interesting open problem.

\begin{figure}[t]
  \centering
  \subfloat[Tree of height 10]{\includegraphics[width=.32\textwidth]{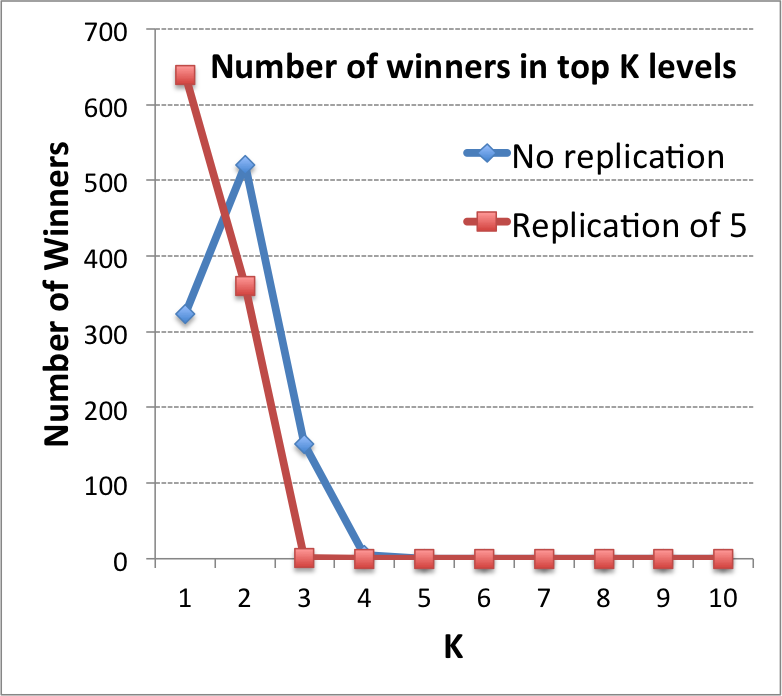}}\hfill
  \subfloat[21x21 Grid]{\includegraphics[width=.29\textwidth]{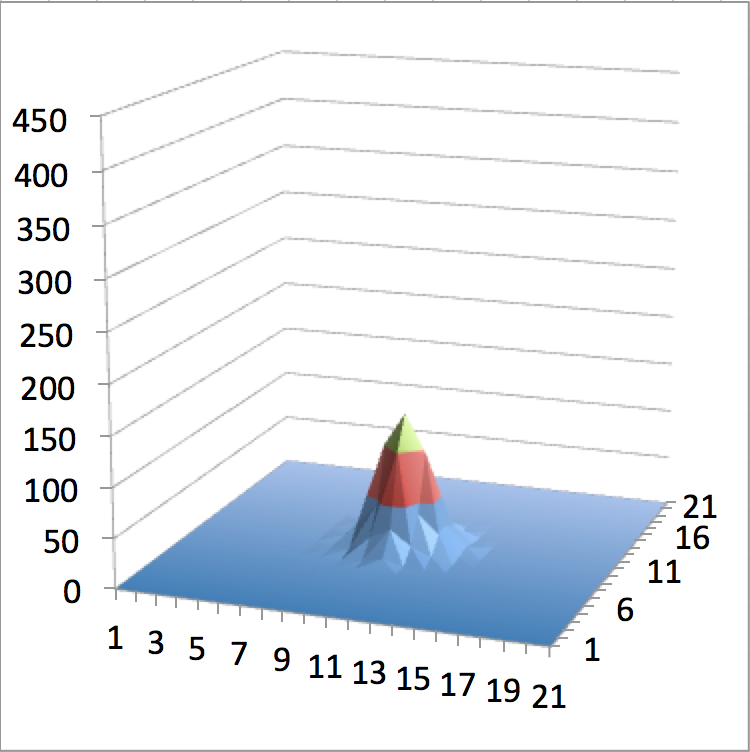}}\hfill
  \subfloat[Grid with 5 initial tokens]{\includegraphics[width=.29\textwidth]{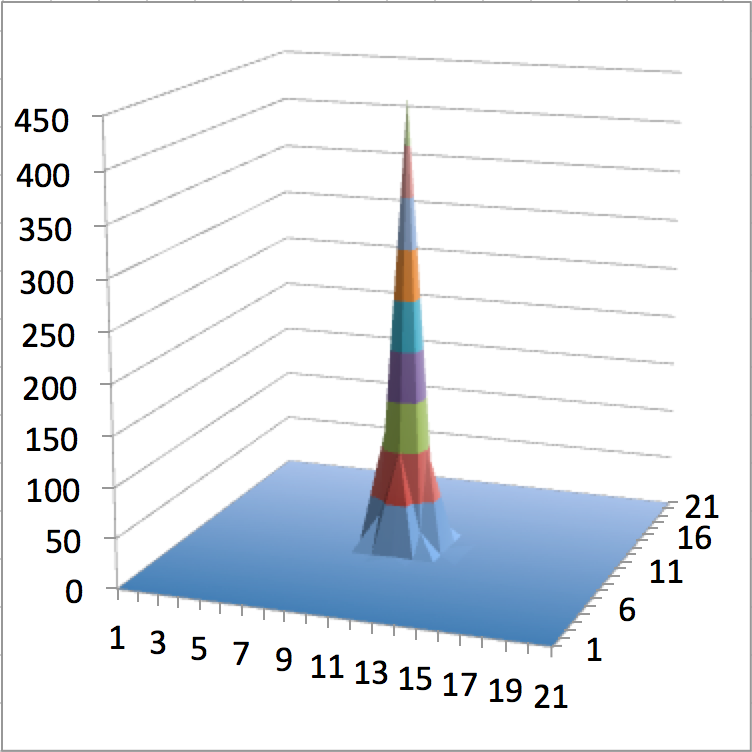}}
  \caption{\small These plots show the number of wins for nodes over 1000 runs of restricted triadic decision-making. We consider two graphs: (a) a binary tree of height 10, and (b, c) a grid of size 21x21. Each of these are simulated with two initial conditions: one or five initial tokens per participant. With five initial tokens, the root node of the tree wins 63.9\% of the time and the top three nodes win 99.9\% of the time. Similarly, with five initial tokens, the center node of the grid wins 44.7\% of the time and the middle 9 nodes win 99.6\% of the time.}
  \label{fig:simulations}
\end{figure}

\begin{theorem}\label{thm:star-restricted}
  Consider $n$ participants on a star graph, for which $j = \Omega(\sqrt{cn\log n})$ of the participants are located at the root, and for which the remaining $n-j$ participants are distributed evenly among the $k$ leaves. Then with probability greater than $1-n^{-c}$, the winning participant $\hat{x}$ is a participant at the root node.
\end{theorem}
\begin{proof}
  We use $\{R_t, t \geq 0\}$, $R_t \in \mathbb{Z}^{k+1}$, to denote a markov chain which is induced by restricted triadic decision-making (RTD) on the star. At time $t$,
  \begin{enumerate}
    \item $R_{t1}$ is the number of tokens at the center (root node) of the star,
    \item $R_{tj}$ is the number of tokens at the $j-1$-th leaf node of the star for $j = 2, \ldots, k+1$.
  \end{enumerate}
  We use $\{R'_t, t \geq 0\}$, $R'_t \in \mathbb{Z}^2$, to denote a truncated representation which encodes the number of tokens at the center of the star and the {\it max} number of tokens at the leaf nodes. Specifically,
  \begin{align*}
    R'_{t1} = R_{t1},\quad\text{ and }\quad R'_{t2} = \max_{2 \leq j \leq k+1}R_{tj}
  \end{align*}
  It is easy to verify that $\{R_t, t \geq 0\}$ is a markov chain because restricted triadic decision-making only depends on the positions of the tokens.
  However, the truncated representation $\{R'_t, t \geq 0\}$ is {\it not} a markov chain since transition probabilities depend on the distribution of the tokens among the leaves.

  A major step in our proof will be to couple $\{R_t, t \geq 0\}$ to another stochastic process $\{S_t, t \geq 0\}$, $S_t \in \mathbb{Z}^{k+1}$, which we will call the concentrated variant of restricted triadic decision-making on the star (CRTD). The state space of $S_t$ is the same as that of $R_t$ and also describes a partition of $n$ tokens over the $k+1$ nodes of the star graph. At time $t$,
  \begin{enumerate}
    \item $S_{t1}$ is the number of tokens at the center (root node) of the star,
    \item $S_{tj}$ is the number of tokens at the $j-1$-th leaf node of the star for $j = 2, \ldots, k+1$.
  \end{enumerate}
  A step of this markov chain involves first taking a step according to the normal RTD transition resulting from a triad. Once this is done, the maximum number of tokens on the leaf nodes is calculated and the tokens on the leaves are then {\it concentrated} while preserving the calculated maximum and ensuring that $S_{t2} \geq S_{t3} \geq \ldots \geq S_{t(k+1)}$. As a simple example, suppose there are $5$, $4$, $3$, and $2$ tokens on leaf nodes $S_{t2}, S_{t3}, S_{t4}, S_{t5}$ respectively after a step of the RTD transition. Then the max number of tokens on the leaf nodes is $5$, and concentrating the tokens would result in $5$, $5$, $4$, and $0$ tokens on leaf nodes $S_{t2}, S_{t3}, S_{t4}, S_{t5}$ respectively. 
  We define a corresponding truncated representation to be $\{S'_t, t \geq 0\}$, $S'_t \in \mathbb{Z}^2$, where
  \begin{align*}
    S'_{t1} = S_{t1},\quad\text{ and }\quad S'_{t2} = \max_{2 \leq j \leq k+1}S_{tj}
  \end{align*}
  It is easy to verify that {\it both} $\{S_t, t \geq 0\}$ and $\{S'_t, t \geq 0\}$ are markov chains.
  With these definitions, we can state the following three lemmas (proof in Appendix \ref{appsec:star-restricted}):
  \begin{enumerate}
    \item (Lemma \ref{lem:crtd-coupling}): Consider any two CRTD dynamics $X_t$, $Y_t$, and their corresponding truncated representations $X'_t$, $Y'_t$. Suppose that $X'_{j1} \leq Y'_{j1}$ and $X'_{j2} \geq Y'_{j2}$ for some $j$. Then we show that a coupling exists such that $X'_{(j+1)1} \leq Y'_{(j+1)1}$ and $X'_{(j+1)2} \geq Y'_{(j+1)2}$.
    \item (Lemma \ref{lem:crtd-rtd-coupling}): Consider a CRTD dynamic $Y_t$ and a RTD dynamic $Z_t$ and their corresponding truncated representations $Y'_t$, $Z'_t$. Suppose that $Y'_{j1} = Z'_{j1}$ and $Y'_{j2} = Z'_{j2}$ for some $j$. Then we show that a coupling exists such that $Y'_{(j+1)1} \leq Z'_{(j+1)1}$ and $Y'_{(j+1)2} \geq Z'_{(j+1)2}$.
    \item (Lemma \ref{lem:crtd-rtd-coupling-full}): Consider a CRTD dynamic $X_t$ and a RTD dynamic $Z_t$ and their corresponding truncated representations $X'_t$, $Z'_t$. Suppose that $X'_{01} \leq Z'_{01}$ and $X'_{02} \geq Z'_{02}$. Then we show that a coupling exists such that $X'_{t1} \leq Z'_{t1}$ and $X'_{t2} \geq Z'_{t2}$ for all $t$ and for every history of the markov chain.
  \end{enumerate}
  We can now wrap up our proof. Consider a CRTD dynamic $X_t$ and a RTD dynamic $Z_t$ and their corresponding truncated representations $X'_t$, $Z'_t$. Suppose that $X_0 = (a, (n-a)/2, (n-a)/2, 0, \ldots, 0)$ and that $X'_{01} \leq Z'_{01}$ and $X'_{02} \geq Z'_{02}$. By Lemma \ref{lem:crtd-rtd-coupling-full}, a coupling exists such that whenever the center of the star wins in $X_t$, i.e. $X'_t$ converges to $(n, 0)$, the center of the star must also win in $Z_t$, i.e. $Z'_t$ has converged to $(n, 0)$. Therefore the probability that the center of the star wins in $Z_t$ is at least as high as the probability that the center of the star wins in $X_t$. 
  
  We now calculate the probability that the center of the star wins in $X_t$. The main observation is that, for our initial condition, the star being considered in $X_t$ only has two leaves. It is easy to verify that in this case, the CRTD dynamic is identical to the RTD dynamic (the leaves are always already ``concentrated''). Moreover, since a star with two leaves is simply a line, we already have the probability that the center of the star wins in $X_t$. For $a = \Omega(\sqrt{cn\log n})$, this probability is greater than $1-n^{-c}$, which concludes our proof.
\end{proof}

%
%
%
%
%
%
%

\section{Triadic decision-making succeeds with strategic agents}\label{sec:triadic-strategic}

One cannot naively apply the prior results to strategic agents in a full LDSG process. This is because a strategic agent wants to ultimately maximize the utility he receives from the global decision $\hat{x}$. Thus, an agent might settle for a worse consensus decision in a given small group if it could help in the long run.
We show that this does not happen and that, in fact, we can achieve the same results described in Theorem \ref{thm:triadic} for strategic agents under a subgame perfect Nash equilibrium~\cite{Shoham2008}. 

\subsection{Game theoretic formalization}

The entire decision-making process induced by Definition \ref{def:tdd} and the majority rule process (Section \ref{sec:triadic-deliberation}) can be formalized as an infinite horizon extensive form game (see Figures \ref{fig:ldsg-game} and \ref{fig:ldsg-extensive-form}).
A full strategy for the game involves choosing, given every possible past history, a strategy for how one votes and proposes motions. For each individual $i$, the utility assigned from the game is defined as $u_i(\hat{x}) = -d(x_i, \hat{x})$ if a winner $\hat{x}$ is produced, and $-\infty$ if the game never terminates.\footnote{We use this definition of utility due to its simplicity and space constraints; however, our results hold for a much broader class of utilities known as single-peaked preferences.}

\begin{figure}[t]
  \centerline{\includegraphics[width=.4\textwidth]{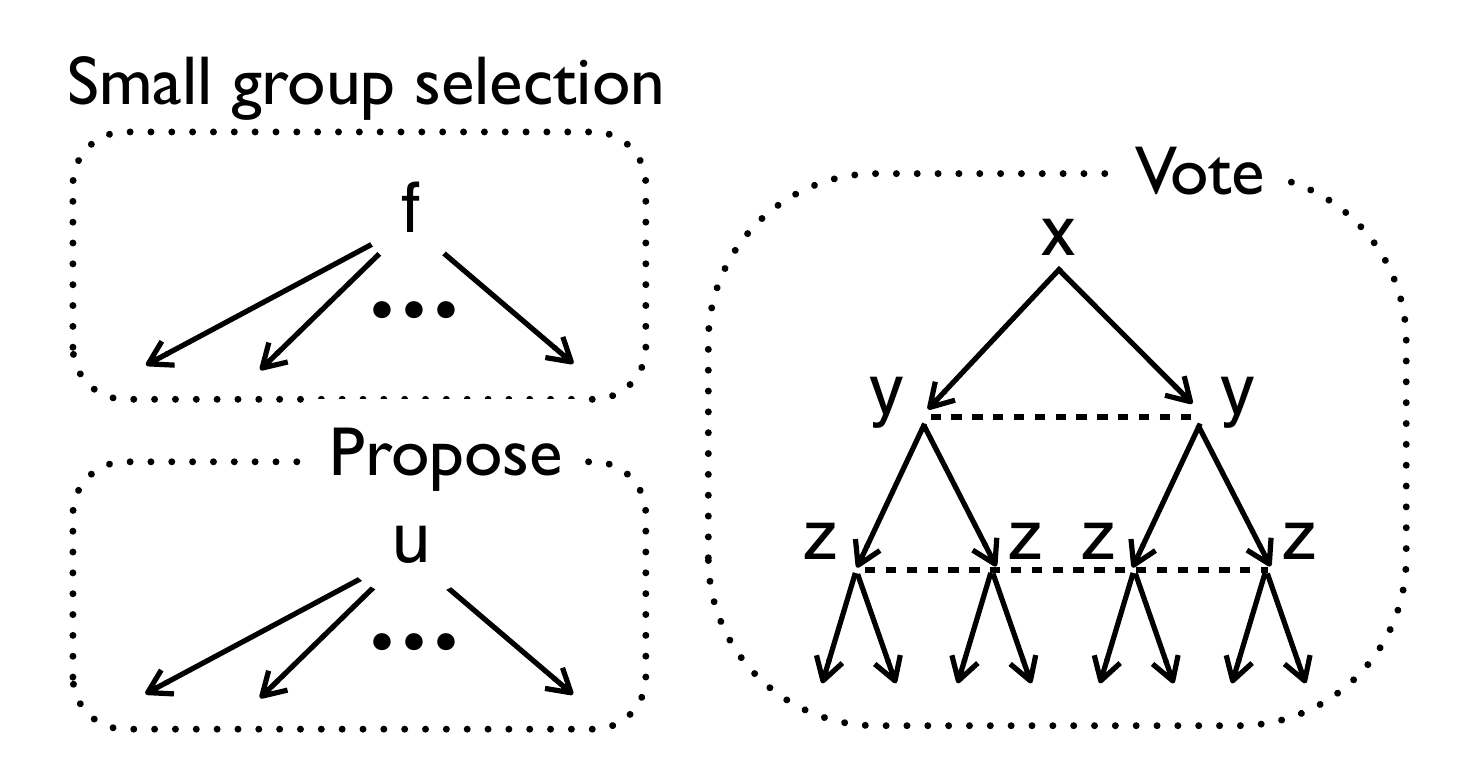}}
  \caption{\small The game-theoretic representation of triadic decision-making with majority rule is an infinite horizon extensive form game made out of several blocks as depicted in Figure \ref{fig:ldsg-game}. The ``Small group selection'' block is played by the small group selection function $f$ who selects one of $n^3$ possible triads according to the probability distribution described by Definition \ref{def:tdd}. The ``Propose'' block for a player $u$ consists of $n+1$ actions corresponding to either proposing to replace the current winner with one of the $n$ participants or motioning to end the process. The ``Vote'' block for the small group $x, y, z$ consists of a simultaneous vote to either accept or reject the proposal. The fact that these votes are simultaneous is represented formally by making the nodes for $y$ and for $z$ a single information set.}
  \label{fig:ldsg-extensive-form}
\end{figure}

\subsection{Approximation and time results for strategic agents in the unrestricted setting}

We show that a strategy which we call truthful bargaining achieves a subgame perfect Nash equilibria for the aforementioned game while also returning the generalized median as the consensus decision of each small group. This means that the results of Theorem \ref{thm:triadic} carry over to this equilibrium.

The truthful bargaining strategy is simple in that it is independent of past small group interactions. Roughly speaking, it requires a participant to vote truthfully between opinions, and to propose his best ``bargaining point'' when one exists, and to motion to end the process otherwise. When voting on a motion to end, truthful bargaining requires a participant to accept the motion when he has no bargaining points, and reject it otherwise.

\begin{definition}
  Suppose that $w$ is the current winner in a triadic majority rule process between participants $x, y, z$. Define the bargaining points for participant $x$ to be any point which is not the current winner $w$, but lies on a shortest path between $w$ and $x$ as well as a shortest path between $w$ and at least one of $y$ or $z$. The best bargaining point for $x$ is the bargaining point which is closest to himself.
\end{definition}

\begin{definition}\label{def:tr}
  Define truthful bargaining to be the strategy in which, independent of the past history,
  \begin{itemize}
    \item If a vote is taken for a proposed alternative, the agent votes truthfully according to his preference,
    \item If it is the agent's turn to propose a motion, he proposes his best bargaining points if one exists, and motions to end otherwise,
    \item If a vote is taken on a motion to end, he accepts if he has no bargaining points, and rejects it otherwise.
  \end{itemize}
  For all of these decisions, ties are broken arbitrarily.
\end{definition}

\begin{theorem}\label{thm:strategic} 
  Consider a triadic decision-making process $L = (V, f, g)$ where participants $V$ form a median graph and $g_t$ is determined by a majority rule process between the participants $S_t$. Then the truthful bargaining strategy is a subgame perfect Nash equilibrium. Moreover, $g_t = \text{generalized-median}(S_t)$ for all $t$.
\end{theorem}
\begin{proof}
We prove this through several lemmas. We first show that, if all participants follow truthful bargaining, then the consensus decision of each triad $S$ will be the generalized median of $S$ (Lemma \ref{lem:truthful-round}). Following this, we consider any triadic round of any subgame in which one participant $u$ deviates. We show that regardless of the strategy that $u$ deviates to, either the round never ends or the generalized median of $S$ lies on a shortest path from $u$ to the resulting winner $w'$ (Lemma \ref{lem:deviate-round}). Finally, we leverage Lemma \ref{lem:deviate-round} to showing that, for any subgame, if one participant $u$ deviates, then the utility received is stochastically dominated by the utility received if $u$ did not deviate (Lemma \ref{lem:deviate-complete}), which concludes our proof.
\end{proof}

Since we have shown that $g_t = \text{generalized-median}(S_t)$, the approximation and time bounds of Theorem \ref{thm:triadic} extend to the strategic case.

\section{Concluding Discussion}\label{sec:discussion}

In this paper, we proposed a model in which sequences of small group interactions are used to scale up decision-making. We showed that small groups can find the wisdom of the crowd efficiently --- sometimes in situations when the group at large is unable to come to a stable decision --- but we also showed that the feasibility of doing so can depend strongly on the size of the small groups. This phenomenon, which also applies to simple opinion formation dynamics, is succinctly summarized by the following aphorism: while two heads are better than one, three heads are better than two.

There are many questions remaining that we view as promising directions for future research. In our paper, we studied a particular approach and dynamic to scaling deliberative decision-making. We believe that the problem of scaling deliberative decision-making deserves a broader and more fundamental framework. Perhaps one can define a model for {\it deliberative social choice}, in which a global choice is calculated from a subset of small group choices. This should be done first for participants whose opinions do not change (like in this paper and in majority rule dynamics), but then should be extended to true deliberative models which account for transfer of information and evolution of opinions. Such mathematical work should also be coupled to systems work which consider motivational challenges and interaction interfaces. Ultimately, all of these should work towards a realization of deliberative democracy in practice.



\begin{acks}
  The authors would like to thank Helene Landemore for discussing and pointing us to relevant deliberative democracy literature, and the reviewers (over the course of several submissions) for pushing us to motivate and communicate our paper more clearly. This work was supported by the Army Research office Grant No. 116388, the Office of Naval Research Grant No. 11904718, and by the Stanford Cyber Initiative. David T. Lee was also supported by an NSF graduate research fellowship and a Brown Institute for Media Innovation Magic Grant.
\end{acks}

\bibliographystyle{ACM-Reference-Format-Journals}
\bibliography{ec16-small-groups_v20}

\newpage

\appendix
\section*{APPENDIX}

\section{Definitions and notation}\label{appsec:definitions}

\begin{definition}\label{def:interval}
  Define the interval $I_{xy}$ between points $x$ and $y$ to be the set containing all points lying on a shortest path between $x$ and $y$, i.e. $I_{xy} = \{w \mid d(x, y) = d(x, w) + d(w, y)\}$. For an edge $e=(i,j)$ we let $e \in I_{xy}$ if $i, j \in I_{xy}$.
\end{definition}

\begin{definition}\label{def:convex}
  Define a set $S$ to be convex if for any two points $x, y$ in $S$, every shortest path between $x$ and $y$ also lies in $S$, i.e. $x, y \in S \implies I_{xy} \in S$. Define the convex hull of a set $S$ to be the smallest convex set that contains $S$. 
\end{definition}

\begin{definition}\label{def:winsets}
  For any edge $e = (u, v)$, define the win sets $W_{uv} = \{w \in V \mid d(w, u) < d(w, v)\}$ and $W_{vu} = \{w \in V \mid d(w, v) < d(w, u)\}$ to be the set of nodes that are closer to $u$ or $v$ respectively. 
\end{definition}

\begin{definition}\label{def:median-of-three}
  We will use the notation $m(x, y, z)$ to denote the generalized median of the points $x, y, z$, i.e. $m(x, y, z) = \arg\min_{u \in V} d(u, x) + d(u, y) + d(u, z)$.
\end{definition}

\begin{definition}\label{def:condorcet}
  A node $x \in V$ is a Condorcet winner for the set $S$ if, for any other node $y \in V$, $\lvert \{ u \in S \mid d(u, x) < d(u, y) \} \rvert > \lvert \{ u \in S \mid d(u, y) < d(u, x) \}\rvert$.
\end{definition}

\begin{definition}\label{def:median-graph}
  A graph is a median graph if, for every $x, y, z$, there is exactly one point which simultaneously lies on some shortest path between $x$ and $y$, between $x$ and $z$, and between $y$ and $z$, i.e. $\lvert I_{xy} \cap I_{xz} \cap I_{yz} \rvert = 1$.
\end{definition}

\section{Properties of median graphs}\label{appsec:median-graphs}

Median graphs are well-studied and have a rich structure. We list several properties that we use, which can be found in the following references~\cite{Knuth2011,Klavzar1999,Mulder1978,Bandelt2008,Imrich2000}.

\begin{lemma}\label{lem:gated-sets}
  Any convex set $S$ in a median graph $V$ is gated. That is, for any $x$ in the graph $V$, there exists some gate $g(x) \in S$ such that for any node $y \in S$, some shortest path from $x$ to $y$ goes through $g(x)$.
\end{lemma}

\begin{lemma}\label{lem:interval-convex}
  The interval $I_{xy}$ between any two nodes $x$ and $y$ in a median graph is convex.
\end{lemma}

\begin{lemma}\label{lem:median-graph-axioms}
  For any $a, b, c, d, e \in V$ where $V$ is a median graph, $m(a, b, m(c, d, e)) = m(m(a, b, c), m(a, b, d), e)$.
\end{lemma}

\begin{lemma}\label{lem:winsets-1}
  For any edge $e = (u, v)$ in a median graph $V$, $W_{uv}$ and $W_{vu}$ are convex sets that partition the nodes.
\end{lemma}

\begin{lemma}\label{lem:winsets-2}
  For any two unique nodes, there exists at least one edge that partitions them.
\end{lemma}

\begin{lemma}\label{lem:winsets-3}
  For a median graph with $n$ nodes, there are at most $\frac{n}{2}\log n$ edges in the graph.
\end{lemma}

\begin{lemma}\label{lem:winsets-4}
  Let $E^*$ denote a minimal set of edges such that the win sets of $E^*$ partition all the nodes of the graph. Then the distance between two nodes $x$ and $y$ is equal to the number of edges in $E^*$ that separate $x$ and $y$.
\end{lemma}

\begin{lemma}\label{lem:winsets-5}
  $x \in I_{yz} \iff x \in W_{uv}$ if $y, z \in W_{uv}$ for all $e =(u,v)$.
\end{lemma}
%

\section{Known random walk lemma}

We use a lemma which is found in \cite{Goel2012}, which discusses how to scale up voting also using the theme of triads.

\begin{lemma}\label{lem:urn-function}
  Consider an absorbing random walk $X_t$ on the integers $0, 1, \ldots, n$ with the following transition probabilities:
  \begin{align*}
    \prob[X_{t+1} = X_t + \Delta] = \begin{cases}
      \frac{3X_t^2(n-X_t)}{n^3} &\text{ if $\Delta = 1$}\\
      \frac{3X_t(n-X_t)^2}{n^3} &\text{ if $\Delta = -1$}\\
      \frac{X_t^3 + (n-X_t)^3}{n^3} &\text{ if $\Delta = 0$}
    \end{cases}
  \end{align*}
  Let $T$ denote the absorption time at which the random walk first hits state $0$ or $n$. Then,
  \begin{align*}
    \prob[X_T = n] &= \left(\frac{1}{2}\right)^{n-1}\sum_{j=1}^n\binom{n-1}{j-1}\\
    \mathbb{E}[T] &\leq n\ln n + O(n)
  \end{align*}
\end{lemma}

\section{Proof of Theorem 6.3}\label{appsec:dyadic-proof}

\begin{proof}[of Theorem \ref{thm:pairwise}]
  Consider the line graph, represented on the reals $\mathcal{R}$, for which $n = 2k+1$ and $x_i = i$ for $i = 1, 2, \ldots, k$ and $x_i = 0$ otherwise. For any group of two participants $S = \{x, y\}$, the convex hull of the participants is simply the interval between them. Clearly, $S$ has the same relationship structure as itself under $\psi(t) = x + y - t$, i.e. the interval ${x, x+1, \ldots, y}$ is isomorphic to itself when reversed. Invoking local consistency implies that
  $g(S) \in \{x, x+1, \ldots, y\}$ and $g(S) = \psi(g(S))$, implying that $g$ must be symmetric about the point $\frac{1}{2}(x + y)$. Define $Y_t = \frac{1}{n}\sum_i y_i^t$ and consider time $t+1$ for which $S_{t+1} = \{j_1, j_2\}$. Then,
  \begin{align*}
    n\mathbb{E}[Y_{t+1} \mid y^t] &= \sum_i \mathbb{E}[y_i^{t+1} \mid y^t]\\
    &= \left(\sum_{i \not\in S_{t+1}}y_i^t\right) + \mathbb{E}[y_{j_1}^{t+1} + y_{j_2}^{t+1} \mid y^t]\\
    &= \left(\sum_{i \not\in S_{t+1}}y_i^t\right) + \mathbb{E}[2\cdot\frac{1}{2}(y_{j_1}^{t} + y_{j_2}^t) \mid y^t] = nY_t
  \end{align*}
  Therefore, $Y_t$ is a martingale, and for any time $t$, $\mathbb{E}[Y_t] = Y_0$. Noting that $D(\cdot)$ is convex, we can apply Jensen's inequality to get $\mathbb{E}\left[\frac{1}{n}\sum_iD(y_i^t)\right] \geq D\left(\mathbb{E}\left[\frac{1}{n}\sum_i y_i^t\right]\right) = D(Y_0)$. It is not hard to verify that $D(Y_0) = (\frac{9}{8}-o(1))D(x^*)$ for the line graph given, which concludes our proof.
\end{proof}

\section{Proof of Theorem 7.2}\label{appsec:median-graph-characterization}

\begin{proof}[of Theorem \ref{thm:median-graph-characterization}]
  We will break this proof up into several lemmas. Consider a set $S = \{x, y, z\}$ and a node $c$. We will refer to the following possible statements for any node $v$
  \begin{enumerate}
    \item $v$ is a Condorcet winner for $S$,
    \item $v$ lies on a shortest path between $x$ and $y$, $x$ and $z$, and $y$ and $z$,
    \item $v$ is a generalized median for $S$.
  \end{enumerate}
  We first show that if $c$ satisfies $1$, then $c$ satisfies $2$ (Lemma \ref{lem:condorcet-implies-shortest}) and is the unique node satisfying $3$ (Lemma \ref{lem:condorcet-implies-unique-median}). But since any node satisfying $2$ must satisfy $3$ (Lemma \ref{lem:shortest-implies-median}), then $c$ must be the unique node satisfying $2$. Therefore, $G$ satisfies $P$ only if $G$ is a median graph (see Definition \ref{def:median-graph}).

  We then show that if $c$ is the unique node satisfying $2$, then $c$ satisfies $1$ (Lemma \ref{lem:unique-shortest-implies-condorcet}). Therefore $G$ satisfies $P$ if $G$ is a median graph (see Definition \ref{def:median-graph}).
\end{proof}

\begin{lemma}\label{lem:condorcet-implies-shortest}
  Consider any three nodes $x, y, z$ of an unweighted and undirected graph. Suppose that there exists a node $c$ which is a Condorcet winner for $x, y, z$. Then $c$ lies on a shortest path between $x$ and $y$, $x$ and $z$, and $y$ and $z$.
\end{lemma}
\begin{proof}
  Suppose for the sake of contradiction that $c$ does not lie on a shortest path between $x$ and $y$, so that 
  \begin{align}
    d(x, c) + d(c, y) > d(x, y)\label{eqn:2-1}
  \end{align}
  Since $c$ is a Condorcet winner, it must beat $y$ in a pairwise election. But since $y$ would clearly vote for itself over $c$, both $x$ and $z$ must vote for $c$. This implies that 
  \begin{align}
    d(x, y) > d(x, c)\label{eqn:2-2}
  \end{align} 
  Now, since the graph is unweighted, there lies a point $p \in I_{xy}$ such that 
  \begin{align}
    d(x, c) = d(x, p)\label{eqn:2-3}
  \end{align} 
  Applying equations (\ref{eqn:2-1}, \ref{eqn:2-2}, \ref{eqn:2-3}) gives us $d(y, p) = d(x, y) - d(x, p) = d(x, y) - d(x, c) < d(y, c)$. Then we have a point $p$ such that $x$ is ambivalent between $p$ and $c$, and $y$ prefers $p$. Then $c$ cannot be a Condorcet winner, and we have a contradiction. Therefore $c$ must lie on a shortest path between $x$ and $y$. The same argument shows that $c$ must also lie on a shortest path between $x$ and $z$, and $y$ and $z$.
\end{proof}

\begin{lemma}\label{lem:condorcet-implies-unique-median}
  Consider any three nodes $x, y, z$ of an unweighted and undirected graph. Suppose that there exists a node $c$ which is a Condorcet winner for $x, y, z$. Then $c$ is the unique generalized median for $x, y, z$.
\end{lemma}
\begin{proof}
  Suppose for the sake of contradiction, that another node $m$ is a generalized median. Since $c$ is a Condorcet winner, we can split the proof into the following two cases.

  {\it Case 1:} None of $x, y, z$ prefer $m$ to $c$, and at least one prefers $c$ to $m$. In this case, it is trivially true that $d(x, c) + d(y, c) + d(z, c) < d(x, m) + d(y, m) + d(z, m)$, so $m$ cannot be a generalized median.

  {\it Case 2:} Two of $x, y, z$ ($x, y$ without loss of generality) prefer $c$ to $m$ and $z$ prefers $m$ to $c$. Since the graph is unweighted, there lies a point $p \in I_{xm}$ such that 
  \begin{align}
    d(x, p) = d(x, c)\label{eqn:2-4}
  \end{align} 
  But then, we must have $d(y, p) \geq d(y, m)$ (Lemma \ref{lem:median-three-distance}). Since $y$ prefers $c$ to $m$, this implies that 
  \begin{align}
    d(y, p) > d(y, c)\label{eqn:2-5}
  \end{align}
  Since $c$ is a Condorcet winner, then equations (\ref{eqn:2-4}, \ref{eqn:2-5}) imply 
  \begin{align}
    d(z, p) \geq d(z, c)\label{eqn:2-6}
  \end{align}
  Applying (\ref{eqn:2-6}), the triangle inequality, the fact that $p \in I_{xm}$, and (\ref{eqn:2-4}) results in $d(z, c) \leq d(z, p) \leq d(z, m) + d(m, p) = d(z, m) + d(m, x) - d(p, x) = d(z, m) + d(m, x) - d(x, c)$. Adding $d(x, c)$ on both sides tells us that 
  \begin{align}
    d(z, c) + d(x, c) \leq d(z, m) + d(x, m)\label{eqn:2-7}
  \end{align}
  Repeating the same argument with $p \in I_{ym}$ gives us 
  \begin{align}
    d(z, c) + d(y, c) \leq d(z, m) + d(y, m)\label{eqn:2-8}
  \end{align}
  But we also know that 
  \begin{align}
    d(x, c) + d(y, c) < d(x, m) + d(y, m)\label{eqn:2-9}
  \end{align}
  since $x$ and $y$ prefer $c$ to $m$. Adding equations (\ref{eqn:2-7}, \ref{eqn:2-8}, \ref{eqn:2-9}) together, we get $2(d(x, c) + d(y, c) + d(z, c)) < 2(d(x, m) + d(y, m) + d(z, m))$, which means that $m$ is not a generalized median and we have a contradiction. Therefore, $c$ must be the unique generalized median.
\end{proof}

\begin{lemma}\label{lem:median-three-distance}
  Consider any three nodes $x, y, z$ of an unweighted and undirected graph and a node $m$ which is a generalized median for $x, y, z$. Then for any point $p \in I_{xm}$, $d(y, p) \geq d(y, m)$.
\end{lemma}
\begin{proof}
  We have
  \begin{align}
    d(x, m) &+ d(y, m) + d(z, m)\notag\\
    &\leq d(x, p) + d(y, p) + d(z, p)\label{eqn:2-10}\\
    &\leq d(x, p) + d(y, p) + d(p, m) + d(z, m)\label{eqn:2-11}\\
    &= d(x, m) + d(y, p) + d(z, m)\label{eqn:2-12}
  \end{align}
  where (\ref{eqn:2-10}) follows from $m$ being the generalized median, (\ref{eqn:2-11}) follows from the triangle inequality, and (\ref{eqn:2-12}) follows from $p\in I_{xm}$. Subtracting $d(x, m) + d(z, m)$ on both sides gives us $d(y, p) \geq d(y, m)$.
\end{proof}

\begin{lemma}\label{lem:shortest-implies-median}
  Consider any three nodes $x, y, z$ of an unweighted and undirected graph. Suppose that there exists a node $m$ which lies on a shortest path between $x$ and $y$, $x$ and $z$, and $y$ and $z$. Then $m$ is a generalized median for $x, y, z$.
\end{lemma}
\begin{proof}
  Consider any other point $m'$. Then,
  \begin{align}
    2(d(x, m) &+ d(y, m) + d(z, m))\notag\\
    &= d(x, y) + d(x, z) + d(y, z)\label{eqn:2-13}\\
    &\leq 2(d(x, m') + d(y, m') + d(z, m'))\label{eqn:2-14}
  \end{align}
  where (\ref{eqn:2-13}) follows from the fact that $m$ lies on the shortest path between every two of $x, y, z$ and (\ref{eqn:2-14}) follows from the triangle inequality.
\end{proof}

\begin{lemma}\label{lem:unique-shortest-implies-condorcet}
  Consider any three nodes $x, y, z$ of an unweighted and undirected graph. Suppose that there exists a unique node $m$ which lies on a shortest path between $x$ and $y$, $x$ and $z$, and $y$ and $z$. Then $m$ is a Condorcet winner for $x, y, z$.
\end{lemma}
\begin{proof}
Consider any other node $m'$. Then there must be one pair ($x$ and $y$ without loss of generality) such that $m'$ does not lie on a shortest path between them. Therefore, $d(x, m) + d(m, y) < d(x, m') + d(m', y)$, which implies that at least one of $d(x, m) < d(x, m')$ or $d(y, m) < d(y, m')$ is true. Let this node be $x$ without loss of generality. If neither of $y$ or $z$ is closer to $m'$ than $m$, then $m$ beats $m'$ in a pairwise election and we are done. 

Now suppose one of $y$ or $z$ is closer to $m'$ than $m$. Then it must be true that the other is closer to $m$ than $m'$ since $d(y, m) + d(m, z) \leq d(y, m') + d(m', z)$. Therefore, $m$ still beats $m'$ in a pairwise election. Since this holds for any $m'$, this means that $m$ is the Condorcet winner.
\end{proof}

\section{Lemmas for Theorem 7.3}


\begin{lemma}\label{lem:edge-approx}
  Consider a triadic decision-making process $L = (V, f, g)$ where participants $V$ form a median graph and $g_t = g(S_t) = \text{generalized-median}(S_t)$. Fix an edge $e = (u, v)$. Then,
  \begin{align*}
    \prob\left[\frac{\lvert \{ i \mid e \in I_{\hat{x}x_i}\} \rvert}{\lvert \{ i \mid e \in I_{x^*x_i}\} \rvert} > \left(1+O\left(\sqrt{\frac{c\log n}{n}}\right)\right)\right] \leq \frac{1}{n^c}
  \end{align*}
\end{lemma}
\begin{proof}
Without loss of generality, suppose that $W_{uv}$ contains less of the $n$ participants and that $W_{vu}$ contains more. Let $N_1$ denote the number of participants in $W_{vu}$ and $N_2$ denote the number of participants in $W_{uv}$, so that $N_2 \leq N_1$ and $N_1 + N_2 = n$ (Lemma \ref{lem:winsets-1}).

We first note that $\lvert \{ i \mid e \in I_{\hat{x}x_i}\} \rvert$ is simply the number of $x_i$ which are in a different partition $W_{uv}$, $W_{vu}$ than $\hat{x}$ (Lemma \ref{lem:winsets-5}). This must be either $N_1$ or $N_2$ depending on whether $\hat{x} \in W_{uv}$ or $\hat{x} \in W_{vu}$, so
\begin{align*}
  \frac{\lvert \{ i \mid e \in I_{\hat{x}x_i}\} \rvert}{\lvert \{ i \mid e \in I_{x^*x_i}\} \rvert} > 1 \iff &\lvert \{ i \mid e \in I_{\hat{x}x_i}\} \rvert = N_1 \text{ and }\\
  &\lvert \{ i \mid e \in I_{x^*x_i}\} \rvert = N_2
\end{align*}
Let us first consider the case when $N_2 \leq \frac{n}{2} - \sqrt{cn \ln n}$. Then we show that $\prob[\hat{x} \in W_{uv}] \leq n^{-c}$ (Lemma \ref{lem:edge-markov-chain}) which implies that $\frac{\lvert \{ i \mid e \in I_{\hat{x}x_i}\} \rvert}{\lvert \{ i \mid e \in I_{x^*x_i}\} \rvert} > 1$ with probability at most $n^{-c}$.

Let us now consider the case when $N_2 \geq \frac{n}{2} - \sqrt{cn \ln n}$. Then, 
\begin{align*}
  \frac{\lvert \{ i \mid e \in I_{\hat{x}x_i}\} \rvert}{\lvert \{ i \mid e \in I_{x^*x_i}\} \rvert} \leq \frac{N_1}{N_2} &= \frac{n- N_2}{N_2} = \left(1+O\left(\sqrt{\frac{c\log n}{n}}\right)\right)
\end{align*}
and our theorem statement is trivially true. 
\end{proof}

\begin{lemma}\label{lem:edge-markov-chain}
  Consider a triadic decision-making process $L = (V, f, g)$ where participants $V$ form a median graph and $g_t = g(S_t) = \text{generalized-median}(S_t)$. Fix an edge $e = (u, v)$ and suppose that less than $\frac{n}{2} - \sqrt{cn\ln n}$ participants belong to the win set $W_{uv}$. Also, let $T_e$ denote the first time $t$ at which all tokens $y_i^t$ belong to either $W_{uv}$ or $W_{vu}$. Then,
  \begin{align}
    &\prob[\hat{x} \in W_{uv}] \leq n^{-c}\label{eqn:edge-prob}\\
    &\prob[T_e > cn\ln^2 n] < O(n^{-c})\label{eqn:edge-time}
  \end{align}
\end{lemma}
\begin{proof}
Let $X_t = \lvert \{ i \mid y_i^t \in W_{uv} \}$ denote the number of tokens in $W_{uv}$ at time $t$. By Lemma \ref{lem:median-majority}, the generalized median of three nodes will belong to $W_{uv}$ if and only if at least two of the three nodes belong to $W_{uv}$. Therefore,
\begin{align*}
    \prob[X_{t+1} = X_t + \Delta] = \begin{cases}
      \frac{3X_t^2(n-X_t)}{n^3} &\text{ if $\Delta = 1$}\\
      \frac{3X_t(n-X_t)^2}{n^3} &\text{ if $\Delta = -1$}\\
      \frac{X_t^3 + (n-X_t)^3}{n^3} &\text{ if $\Delta = 0$}
    \end{cases}
  \end{align*}
  To prove Equation \ref{eqn:edge-prob}, we note that the event $\hat{x} \in W_{uv}$ corresponds to the event that $X_{T_e} = n$. Applying Lemma \ref{lem:urn-function}, we have
  \begin{align}
    \prob[\hat{x} \in W_{uv}] &\leq \left(\frac{1}{2}\right)^{n-1}\sum\limits_{j = 1}^{\frac{n}{2}-\sqrt{cn\ln n}}\binom{n-1}{j-1}\notag\\
    &\leq \text{exp}\left(-\frac{1}{2}\cdot\frac{n}{2}\cdot\left(\frac{\sqrt{cn\ln n}}{n/2}\right)^2\right) = \frac{1}{n^c}\notag
  \end{align}
  where the second inequality is found by interpreting the binomial expression in terms of coin flips and applying Chernoff's bound. 

  Equation \ref{eqn:edge-time} follows by applying techniques for high-probability bounds for probabilistic recurrences (we follow the technique of \cite{AspnesNotes}):

  Let $T(i)$ denote $T_e$ given $X_0 = i$. Then $T(i)$ satisfies the probabilistic recurrence
  \[T(i) = p_iT(i+1) + q_iT(i-1) + r_iT(i) + 1\]
  where $p_i, q_i, r_i$ are the probabilities of the transitions $i\to i+1$, $i\to i-1$ and $i\to i$. From Lemma \ref{lem:urn-function}, we have that $\mathbb{E}[T(i)] \leq T^*$ for all $i$, and for $T^* = n\ln n + O(n)$.
  By Markov's inequality, $\prob[T(i) \geq \alpha T^*] \leq \frac{1}{\alpha}$, for any $\alpha > 0$. By conditional probability, $\prob[T(i) \geq 2\alpha T^*] = \prob[T(i) \geq 2\alpha T^* \mid T(i) \geq \alpha T^* ]\prob[T(i) \geq \alpha T^*]$. But note that $\prob[T(i) \geq 2\alpha T^* \mid T(i) \geq \alpha T^*] \leq \frac{1}{\alpha}$ since the remaining random walk conditioned at the point when $T(i) = \alpha T^*$ is just distributed as $T(i')$ for some $i'$. Repeating this logic gives $\prob[T(i) \geq k\alpha T^*] \leq \alpha^k$. Choosing $\alpha = e$ and $k = c\ln n$ gives us our desired result. Note that the point at which $T(i) = \alpha T^*$ exists for any $T^*$ chosen such that $\alpha T^*$ is an integer since $T(i)$ increases in integer increments. 
\end{proof}

\begin{lemma}\label{lem:median-majority}
  Consider any edge $e = (u, v)$ of a median graph $V$ and three nodes $S = \{x, y, z\}$. Let $W^*$ denote the set $W_{uv}$ or $W_{vu}$ which contains at least two of the three nodes. Then $\text{generalized-median}(S) \in W^*$.
\end{lemma}
\begin{proof}
First we note that $W_{uv}$ and $W_{vu}$ partition $V$ (Lemma \ref{lem:winsets-1}). Therefore, by pigeonhole principle, at least two of $x, y, z$ belong to the same set, so that $W^*$ is well defined. Suppose, without loss of generality, that $x, y \in W^*$. We know that $W^*$ is convex (Lemma \ref{lem:winsets-1}). Therefore, $I_{xy} \subseteq W^*$ (Definition \ref{def:convex}). But then $m(x, y, z) \in I_{xy} \subseteq W^*$.
%
\end{proof}


\section{Lemmas for Theorem 7.6}\label{appsec:star-restricted}

\begin{lemma}\label{lem:crtd-coupling}
  Consider any two CRTD dynamics $X_t$, $Y_t$, and their corresponding truncated representations $X'_t$, $Y'_t$. Suppose that $X'_{j1} \leq Y'_{j1}$ and $X'_{j2} \geq Y'_{j2}$ for some $j$. Then a coupling exists such that $X'_{(j+1)1} \leq Y'_{(j+1)1}$ and $X'_{(j+1)2} \geq Y'_{(j+1)2}$.
\end{lemma}
\begin{proof}
  The proof for this is just a lot of algebra, so we will not go through it here. Since $X'_t$, $Y'_t$ are markov chains, we can simply write out the expressions for all the state transitions and their probabilities. Once we do so, breaking it up into multiple cases will enable one to verify that a coupling of these different cases satisfies the desired property (that $X'_{(j+1)1} \leq Y'_{(j+1)1}$ and $X'_{(j+1)2} \geq Y'_{(j+1)2}$).
\end{proof}

\begin{lemma}\label{lem:crtd-rtd-coupling}
  Consider a CRTD dynamic $Y_t$, a RTD dynamic $Z_t$, and their corresponding truncated representations $Y'_t$, $Z'_t$. Suppose that $Y'_{j1} = Z'_{j1}$ and $Y'_{j2} = Z'_{j2}$ for some $j$. Then a coupling exists such that $Y'_{(j+1)1} \leq Z'_{(j+1)1}$ and $Y'_{(j+1)2} \geq Z'_{(j+1)2}$.
\end{lemma}
\begin{proof}
  The proof for this is almost identical for the prior case. The only difference is that $Z'_t$ is not a markov chain. In order to compare to $Y'_t$, simply use $Z_t$ and write expressions for the probabilities that we get transitions in $Z'_t$ (in terms of $Z_t$). We will be able to bound certain expressions through the following optimization problem
  \begin{equation*}
  \begin{aligned}
    & \underset{z_2,\ldots,z_{k+1}}{\text{maximize}}
    & & \sum_{2 \leq j \leq k+1}\left(\frac{z_j}{n}\right)^2 \\
    & \text{subject to}
    & & \sum_{2 \leq j \leq k+1}z_j = \sum_{2 \leq j \leq k+1}Z_{tj},\\
    &&& 0 \leq z_j \leq Z'_{t2}, \; j = 2, \ldots, k+1,\\
    &&& z_2 \geq z_3 \geq \dots \geq z_{k+1}.
  \end{aligned}
  \end{equation*}
  which is solved at
  \begin{align*}
    z_j = \begin{cases}
      Z'_{t2} &\text{ for $j \leq \lfloor\sum_{i=2}^{k+1}Z_{ti}/Z'_{t2}\rfloor$}\\
      \sum_{i=2}^{k+1}Z_{ti}/Z'_{t2} \bmod Z'_{t2} &\text{ for $j = \lfloor\sum_{i=2}^{k+1}Z_{ti}/Z'_{t2}\rfloor + 1$}\\
      0 &\text{ otherwise}
    \end{cases}
  \end{align*}
  Once this is done, breaking it up into multiple cases will enable one to verify that a coupling of these different cases satisfies the desired property (that $Y'_{(j+1)1} \leq Z'_{(j+1)1}$ and $Y'_{(j+1)2} \geq Z'_{(j+1)2}$).
\end{proof}

\begin{lemma}\label{lem:crtd-rtd-coupling-full}
  Consider a CRTD dynamic $X_t$, a RTD dynamic $Z_t$, and their corresponding truncated representations $X'_t$, $Z'_t$. Suppose that $X'_{01} \leq Z'_{01}$ and $X'_{02} \geq Z'_{02}$. Then a coupling exists such that $X'_{t1} \leq Z'_{t1}$ and $X'_{t2} \geq Z'_{t2}$ for all $t$ and for every history of the markov chain.
\end{lemma}
\begin{proof}
  Consider some time $j$ such that $X'_{j1} \leq Z'_{j1}$ and $X'_{j2} \geq Z'_{j2}$. Define a CRTD dynamic $Y_t$ and its corresponding truncated representation $Y'_t$. Let $Y_j = Z_j$ so that $Y'_{j1} = Z'_{j1}$ and $Y'_{j2} = Z'_{j2}$. We then have $X'_{j1} \leq Y'_{j1} = Z'_{j1}$ and $X'_{j2} \geq Y'_{j2} = Z'_{j2}$. 
  
  By Lemmas \ref{lem:crtd-coupling} and \ref{lem:crtd-rtd-coupling}, it is possible to define a coupling such that $X'_{(j+1)1} \leq Y'_{(j+1)1} \leq Z'_{(j+1)1}$ and $X'_{(j+1)2} \geq Y'_{(j+1)2} \geq Z'_{(j+1)2}$. Therefore, we have proven the inductive hypothesis that $X'_{j1} \leq Z'_{j1}$ and $X'_{j2} \geq Z'_{j2}$ $\implies$ $X'_{(j+1)1} \leq Z'_{(j+1)1}$ and $X'_{(j+1)2} \geq Z'_{(j+1)2}$. Since we have the condition satisfied for $j=0$ based on the initial condition, we can simply use the inductive step to conclude our result for all $t$.
\end{proof}

\section{Formal strategic definitions}

\begin{definition}\label{def:tmr}
  A triadic majority rule (TMR) process is the tuple $M = (V, S, a, v)$, where $V$ is the set of all participants, $S = \{x, y, z\} \subseteq V$ is the set of three participants in the small group, $\{a_t, t \geq 1\}$ is the proposal made during step $t$, which can either be one of the participants in $V$ or $\emptyset$ for a motion to end, and $\{v_t(u), t \geq 1, u \in S\}$ denotes the vote cast by participant $u$ during step $t$, which can either be $1$ or $0$ for accepting or rejecting the proposal respectively.
  Let $w_t$ denote the winner at the beginning of step $t$, $p_t$ denote the participant proposing during step $t$, $v_t^*$ denote the majority vote in step $t$, $T$ denote the step at which the round ends, and $\hat{w}$ denote the winner of the majority rule process. Then,
  \begin{align*}
    w_1 &= x\\
    p_1 &= y\\
    v_t^* &= \text{majority}(v_t(x), v_t(y), v_t(z))\\
    w_{t+1} &= \begin{cases}
      a_t \text{ if $a_t \neq\emptyset$ and $v_t^* = 1$}\\
      w_t \text{ otherwise}
    \end{cases}\\
    p_{t+1} &= \begin{cases}
      u \text{ if $v_t(u) \neq v_t^*$}\\
      p_t \text{ otherwise}
    \end{cases}\\
    T &= \min\{ t \mid a_t = \emptyset \text{ and } v_t^* = 1\}\\
    \hat{w} &= w_T
  \end{align*}
\end{definition}

\begin{definition}\label{def:preferred-points}
  During step $t$ of a TMR process, define the ``preferred points'' for a participant $u$ to be the set $P_u^t$ of points which are on a shortest path from $u$ to the current winner $w_t$, but not equal to $w_t$, i.e. $P_u^t = I_{uw_t} \setminus \{w_t\}$.
\end{definition}
 
\begin{definition}\label{def:bargaining-points}
  During step $t$ of a TMR process, define the ``bargaining points'' for the participant $u$ to be the set $B_u^t$ of points which are preferred points for both $u$ and at least one of the other participants $u'$ or $u''$, i.e. $B_u^t = (P_u^t \cap P_{u'}^t) \cup (P_u^t \cap P_{u''}^t)$.
\end{definition}
 
\begin{definition}
  During step $t$ of a TMR process, define the ``best bargaining point'' $b_u^{*t}$ for the participant $u$ to be the closest bargaining point to $u$, i.e. $b_u^{*t} = \arg\min_{b \in B_u^t} d(b, u)$.
\end{definition}


\begin{definition}\label{def:tr-formal}
  Define truthful bargaining to be the strategy in which agent $x$ proposes and votes in the following way:
  \begin{align*}
    a_t &= \begin{cases}
      b_x^{*t} \text{ if $B_x^t \neq \emptyset$}\\
      \emptyset \text{ otherwise}
    \end{cases}\\
    v_t(x) &= \begin{cases}
      1 \text{ if $a_t = \emptyset$ and $B_x^t = \emptyset$}\\
      1 \text{ if $a_t \neq \emptyset$ and $d(x, a_t) < d(x, w_t)$}\\
      0 \text{ otherwise }
    \end{cases}
  \end{align*}
 \end{definition}


\section{Lemmas for Theorem 8.3}

\begin{lemma}\label{lem:truthful-round}
  Consider a round of TMR in which the three participants $x, y, z$ follow truthful bargaining. Then the winner of the round $\hat{w}$ will be the generalized median of $x, y, z$.
\end{lemma}
\begin{proof}
As defined in Definition \ref{def:tmr}, $w_1 = x, p_1 = y$. Since $y$ is following truthful bargaining, he proposes his best bargaining point, which is the closer of $m(y, x, w_1)$ and $m(y, z, w_1)$ (Lemma \ref{lem:best-bargaining}). Then since $m(y, x, w_1) = m(y, x, x) = x$, and since $m(y, z, w_1) = m(y, z, x)$ lies on a shortest path from $y$ to $x$ (Theorem \ref{thm:median-graph-characterization}), we have $a_1 = m(x, y, z)$. We now have the votes $v_1(x) = 0,\quad v_1(z) = 1$ since $x$ clearly prefers himself to any other point, and $z$ prefers $m(x, y, z)$ since $m(x, y, z)$ must lie on a shortest path from $z$ to $w_1 = x$ (Theorem \ref{thm:median-graph-characterization}). This implies $v_1^* = 1$. Therefore,
  \begin{align*}
    w_2 = m(x, y, z), \quad p_2 = x
  \end{align*}
  It can be shown that none of the participants have any bargaining points at this step (Lemma \ref{lem:all-no-bargaining}), and thus, $x$ motions to end, and all participants vote to accept, i.e.
  \begin{align*}
    a_2 = \emptyset,\quad v_2(y) = 1,\quad v_2(z) = 1
  \end{align*}
  Therefore, $T = 2$ and $\hat{w} = w_T = m(x, y, z)$.
\end{proof}

\begin{lemma}\label{lem:best-bargaining}
  At step $t$ of a TMR round, the best bargaining point of $y$, if it exists, is the closer of $m(y, x, w_t)$ and $m(y, z, w_t)$.
\end{lemma}
\begin{proof}
We know that $B_y^t = (I_{w_tm(w_t, y, x)} \cup I_{w_tm(w_t, y, z)}) \setminus \{w_t\}$ (Lemma \ref{lem:bargaining-points}). But we also know that for any $u \in I_{w_tm(w_t, y, x)}$, $m(w_t, y, x) \in I_{yu}$ (Lemma \ref{lem:median-separates}). Therefore, $m(w_t, y, x)$ is the closest of all points in $I_{w_tm(w_t, y, x)}$ to $y$. One can argue similarly that $m(w_t, y, z)$ is the closest of all points in $I_{w_tm(w_t, y, z)}$ to $y$, and our result follows.
\end{proof}

\begin{lemma}\label{lem:bargaining-points}
  At step $t$ of a TMR round, the bargaining points for $y$ are exactly those points which lie on a shortest path from $m(w_t, y, x)$ to $w_t$ or a shortest path from $m(w_t, y, z)$ to $w_t$, but not including $w_t$, i.e. $B_y^t = (I_{w_tm(w_t, y, x)} \cup I_{w_tm(w_t, y, z)}) \setminus \{w_t\}$.
\end{lemma}
\begin{proof}
By Definitions \ref{def:preferred-points} and \ref{def:bargaining-points}, $B_y^t = (P_y^t \cap P_x^t) \cup (P_y^t \cap P_{z}^t) = ((I_{yw_t} \cap I_{xw_t}) \cup (I_{yw_t} \cap I_{zw_t})) \setminus \{w_t\}$. Then our result follows from Lemma \ref{lem:interval-intersection} since $I_{w_tm(w_t, y, x)} = I_{xw_t} \cap I_{yw_t}$ and $I_{w_tm(w_t, y, z)} = I_{yw_t} \cap I_{zw_t}$.
\end{proof}

\begin{lemma}\label{lem:interval-intersection}
  For nodes $x, y, z$ in a median graph, $I_{xy} \cap I_{xz} = I_{xm(x,y,z)}$.
\end{lemma}
\begin{proof}
First, consider any $u \in I_{xm(x, y, z)}$. Then we can apply Lemma \ref{lem:median-separates} to get $u \in I_{xy}$ and $u \in I_{xz}$. Now consider any $v \in I_{xy} \cap I_{xz}$. Then we know that $m(v, x, y) = m(v, x, z) = v$ and applying Lemma \ref{lem:median-graph-axioms}, gets
\begin{align*}
  m(v, x, m(x, y, z)) &= m(m(v, x, y), m(v, x, z), x)\\
  &= m(v, v, x) = v
\end{align*}
Therefore, $v \in I_{xm(x, y, z)}$ and we get our result.
\end{proof}

\begin{lemma}\label{lem:median-separates}
  Consider $x, y, z \in V$, where $V$ is a median graph. Suppose for $w \in V$, $w$ is on a shortest path from $y$ to $m(x, y, z)$. Then, $m(x, y, z)$ must be on a shortest path from $x$ to $w$ and $w$ must be on a shortest path from $x$ to $y$.
\end{lemma}
\begin{proof}
By triangle inequality, $d(x, w) \leq d(x, m(x, y, z)) + d(m(x, y, z), w)$. But we also have
\begin{align}
  d(x, w) &\geq d(x, y) - d(y, w)\label{eqn:4-1}\\
          &= d(x, m(x, y, z)) + d(m(x, y, z), y) - d(y, w)\label{eqn:4-2}\\
          &= d(x, m(x, y, z)) + d(w, m(x, y, z))\label{eqn:4-3}
\end{align}
where (\ref{eqn:4-1}) follows from the triangle inequality, (\ref{eqn:4-2}) from $m(x,y,z) \in I_{xy}$ (Theorem \ref{thm:median-graph-characterization}), and (\ref{eqn:4-3}) from $w \in I_{ym(x,y,z)}$ (lemma statement). Therefore, we must have equality which means $m(x,y,z) \in I_{xw}$. A similar argument shows us that $w \in I_{xy}$.
\end{proof}

\begin{lemma}\label{lem:all-no-bargaining}
  Suppose that, at step $t$ of a TMR round, $w_t = m(x, y, z)$. Then no participants have any bargaining points.
\end{lemma}
\begin{proof}
This follows directly from Lemma \ref{lem:no-bargaining}.
\end{proof}

\begin{lemma}\label{lem:no-bargaining}
Participant $y$ does not have any bargaining points at step $t$ of a TMR round if and only if $w_t \in I_{ym(x, y, z)}$.
\end{lemma}
\begin{proof}
\begin{align}
  B_y^t = \emptyset &\iff (I_{w_tm(w_t, y, x)} \cup I_{w_tm(w_t, y, z)}) = \{w_t\}\label{eqn:4-4}\\
&\iff m(w_t, y, x) = m(w_t, y, z) = w_t\notag\\
&\iff w_t \in I_{yx} \text{ and } w_t \in I_{yz}\notag\\
&\iff w_t \in I_{ym(x, y, z)}\label{eqn:4-5}
\end{align}
where (\ref{eqn:4-4}) is from Lemma \ref{lem:bargaining-points} and (\ref{eqn:4-5}) is from Lemma \ref{lem:interval-intersection}.
%
\end{proof}

\begin{lemma}\label{lem:deviate-round}
  Consider any subgame of a triadic decision-making process $L = (V, f, g)$ where participants $V$ form a median graph and $g_t$ is determined by a majority rule process between the participants $S_t$. Consider any round $s$ of TMR within the subgame, including the round in which the initial node lies, and suppose that all participants follow the truthful bargaining strategy except for one deviating participant $x \in S_s$. Then either the round of TMR never ends or the generalized median of $S_s$ lies on a shortest path between $x$ and the winner of the $s$-round $\hat{w} = g_s$.
\end{lemma}
\begin{proof}
Let $S_s = \{x, y, z\}$. We know that the initial node of a subgame cannot be in the middle of a vote since all participants vote simultaneously (as captured formally by the information sets). 
Suppose that the TMR round eventually ends. Then we know that two of the three participants must have accepted a motion to end (Definition \ref{def:tmr}). This implies that one of $y, z$ must have accepted a motion to end. Since they are both following truthful bargaining, this means that at least one of them, $y$ without loss of generality, has no bargaining points (Definition \ref{def:tr-formal}). But the only way for $y$ not to have any bargaining points is if $\hat{w}$ is on a shortest path from $y$ to $m(x, y, z)$ (Lemma \ref{lem:no-bargaining}). But from here, it follows easily that $m(x, y, z)$ must be on a shortest path from $x$ to $\hat{w}$ (Lemma \ref{lem:median-separates}).
\end{proof}

\begin{lemma}\label{lem:deviate-complete}
  Consider any subgame of a triadic decision-making process $L = (V, f, g)$ where participants $V$ form a median graph and $g_t$ is determined by a majority rule process between the participants $S_t$. Suppose that all participants follow the truthful bargaining strategy except for one deviating participant $x$. Then the utility received by $x$ is stochastically dominated by the utility received if $x$ did not deviate.
\end{lemma}
\begin{proof}
Fix any subgame. Let $s$ denote the current TMR round that this subgame is a part of. Recall that $y_i^t$ denotes the participant holding the $i$-th token after the $t$-th round. Since $x$ does not follow the truthful bargaining strategy, we will refer to the resulting sequence of $y_i^t$ as the ``deviating process''. 

Now consider another instance of the same subgame in which all participants including $x$ follow the truthful bargaining strategy. Let $\hat{y}_i^t$ denote the participant holding the $i$-th token after the $t$-th round in this instance. We will refer to the resulting sequence of $\hat{y}_i^t$ as the ``truthful process''.

Since both instances start at the same subgame, the tokens are distributed identically, and we can choose the indices such that $\hat{y}_i^s = y_i^s$ for all $i$. Now consider the evolution of these two processes but, crucially, couple them so that the tokens chosen in each round $I_t$ are identical in both processes. In other words, if $y_i^t, y_j^t, y_k^t$ are chosen for the $t+1$-th round in the deviating process, then $\hat{y}_i^t, \hat{y}_j^t, \hat{y}_k^t$ are chosen for the $t+1$-th round in the truthful process.

Clearly, at round $s$, $\hat{y}_i^s \in I_{xy_i^s}$ for all $i$ since $\hat{y}_i^s = y_i^s$. We then give an inductive argument: if $\hat{y}_i^t \in I_{xy_i^t}$ for all $i$, then either $\hat{y}_i^{t+1} \in I_{xy_i^{t+1}}$ for all $i$ or the deviating process never ends (Lemma \ref{lem:inductive-domination}). Therefore, either the deviating process never ends or, for each $i$ and $t \geq s$, $\hat{y}_i^{t}$ is always on a shortest path from $x$ to $y_i^t$. Clearly, if the deviating process never ends, the utility received by $x$ in the truthful process is greater. If the deviating process does end, then $\hat{y}_i^T$ is on a shortest path from $x$ to $y_i^T$ and therefore, $u_x(\hat{y}_i^T) \geq u_x(y_i^T)$ for all $i$. It follows trivially that the utility received by $x$ in the truthful process stochastically dominates the utility received by $x$ in the deviating process.
\end{proof}

\begin{lemma}\label{lem:inductive-domination}
  Consider the truthful and deviating processes described in the proof of Lemma \ref{lem:deviate-complete}. Let $\hat{y}_i^t$ and $y_i^t$ denote the tokens of the truthful and deviating processes respectively. Suppose that after the $t$-th round, $\hat{y}_i^t \in I_{xy_i^t}$ for all $i$. Then after the $t+1$-th round, either $\hat{y}_i^{t+1} \in I_{xy_i^{t+1}}$ for all $i$ or the deviating process never ends.
\end{lemma}
\begin{proof}
Let $\hat{S}_t$ and $S_t$ denote the small groups of the truthful and deviating processes respectively. Let $I_{t+1} = \{i, j, k\}$, so that $\hat{S}_{t+1} = \{\hat{y}_i^{t}, \hat{y}_j^t, \hat{y}_k^t\}$ and $S_{t+1} = \{y_i^t, y_j^t, y_k^t\}$. We will consider three separate cases.

{\it Case 1: $x$ is represented twice in $S_{t+1}$.} By our assumption that $\hat{y}_i^t \in I_{xy_i^t}$, it must be true that $x$ is also represented twice in $\hat{S}_{t+1}$. Then $g(\hat{S}_{t+1}) = x$, so we have our result trivially regardless of what $g(S_{t+1})$ is.

{\it Case 2: $x$ is not represented in $S_{t+1}$.} Since $x$ is not represented in $S_{t+1}$, then all agents $y_i^t, y_j^t, y_k^t$ follow truthful bargaining. By the lemma assumption, all agents $\hat{y}_i^t, \hat{y}_j^t, \hat{y}_k^t$ follow truthful bargaining. Therefore, $g(S_{t+1}) = \text{generalized-median}(S_{t+1})$ (Lemma \ref{lem:truthful-round}) and $g(\hat{S}_{t+1}) = \text{generalized-median}(\hat{S}_{t+1})$ (Lemma \ref{lem:truthful-round}). Then our result follows from Lemma \ref{lem:median-preserves-domin}.
  
{\it Case 3: $x$ is represented once in $S_{t+1}$.} Let $m$ denote $g(S_{t+1})$ if $x$ follows truthful bargaining. Then if $x$ deviates, either the process never ends, or the winner $m' = g(S_{t+1})$ satisfies $m \in I_{xm'}$ (Lemma \ref{lem:deviate-round}). By the lemma assumption, all agents $\hat{y}_i^t, \hat{y}_j^t, \hat{y}_k^t$ follow truthful bargaining, so that $g(\hat{S}_{t+1}) = \text{generalized-median}(\hat{S}_{t+1})$. Then $g(\hat{S}_{t+1}) \in I_{xm}$ (Lemma \ref{lem:median-preserves-domin}). Therefore, either the deviating process never ends or $g(\hat{S}_{t+1}) \in I_{xm'}$.
\end{proof}

\begin{lemma}\label{lem:median-preserves-domin}
  Consider $x, y, z, x', y', z', u \in V$, where $V$ is a median graph and suppose that $x \in I_{ux'}, y \in I_{uy'}, z \in I_{uz'}$. Then, $m(x, y, z) \in I_{um(x', y', z')}$.
\end{lemma}
\begin{proof}
 Fix any edge $e \in V$ and consider the partitions induced by its win sets (see Lemma \ref{lem:winsets-1}). By Lemma \ref{lem:median-majority}, we know that $m(x', y', z')$ is in the same win set as $u$ if and only if at least two of $x', y', z'$ are in the same win set as $u$. By Lemma \ref{lem:winsets-5}, we know that this must imply that at least two of $x, y, z$ are in the same win set as $u$, which means that $m(x, y, z)$ must also be in the same win set as $u$ (Lemma \ref{lem:median-majority}). 
  Therefore, for every edge, whenever $m(x', y', z')$ is in the same win set as $u$, $m(x, y, z)$ is also in the same win set. By Lemma \ref{lem:winsets-5}, this implies that $m(x, y, z)$ lies on a shortest path from $u$ to $m(x', y', z')$, which concludes our proof.
\end{proof}


\end{document}